\documentclass[11pt]{article}

\usepackage[round]{natbib}
\usepackage[margin=1in]{geometry}

\usepackage{microtype}
\usepackage{textcomp}
\usepackage[normalem]{ulem}

\usepackage{amsmath}
\usepackage{amssymb}
\usepackage{amsfonts}
\usepackage{mathtools}
\usepackage{amsthm}
\usepackage{dsfont}
\usepackage{bm}

\usepackage{graphicx}
\usepackage{xcolor}
\usepackage{booktabs}
\usepackage{longtable}
\usepackage{wrapfig}
\usepackage{rotating}
\usepackage{float}
\usepackage{caption}
\usepackage{subcaption} 

\usepackage{enumitem}
\usepackage{comment}
\usepackage{appendix}
\usepackage{authblk}
\usepackage[textsize=tiny]{todonotes}

\usepackage{tikz}
\usepackage{tikz-3dplot}
\usetikzlibrary{shapes.multipart, shapes.symbols}
\usepackage{pgfplots}
\pgfplotsset{compat=1.18}

\usepackage{algorithm}
\usepackage{algorithmic}

\theoremstyle{plain}
\newtheorem{theorem}{Theorem}[section]

\newtheorem{lemma}{Lemma}[section]
\newtheorem{corollary}{Corollary}[section]
\theoremstyle{definition}
\newtheorem{definition}{Definition}

\theoremstyle{remark}

\newcommand{\veps}{{\varepsilon}}
\newcommand{\view}[1]{\text{View}_{\mathcal M}(#1)}

\usepackage{url}
\usepackage[hidelinks,colorlinks,citecolor=black,linkcolor=black,urlcolor=black]{hyperref}

\title{Differentially Private Spectral Graph Clustering: \\ Balancing Privacy, Accuracy, and Efficiency}

\author[1]{Antti~Koskela$^{*}$}
\author[2]{Mohamed~Seif$^{*}$}
\author[2]{H. Vincent~Poor}
\author[2,3]{Andrea~J. Goldsmith}

\affil[1]{\small
Nokia Bell Labs}

\affil[2]{\small
Department of Electrical and Computer Engineering, Princeton University}

\affil[3]{\small
Stony Brook University}

\date{} 

\makeatletter
\renewcommand\AB@affilsepx{ \protect\\[0.2em]} 
\makeatother

\affil[ ]{\small $^{*}$These authors contributed equally to this work.}

\begin{document}

\maketitle

\begin{abstract}

We study spectral graph clustering under edge differential privacy. We propose a matrix shuffling mechanism that combines randomized edge flipping with a random permutation of the adjacency matrix. While edge flipping alone provides only a constant $\varepsilon$ guarantee as the graph grows, shuffling amplifies privacy so that the effective $\varepsilon$ tends to zero with the number of nodes. We develop a unified error analysis framework---based on Davis--Kahan perturbation theory and a classification-margin bound---that gives explicit misclassification rates for all the  mechanisms considered as a function of the privacy budget, eigengap, and number of communities. Applying this framework, we show that the matrix shuffling mechanism achieves an error rate scaling of $\tilde{O}(1/n)$, a clear improvement over two canonical DP baselines from the private PCA literature: the Gaussian mechanism applied directly to the adjacency matrix (Analyze Gauss) and the noisy power method, both of which scale as $\tilde{O}(1)$ in $n$. We further propose a private spectral gap detection algorithm for estimating the number of communities. Experiments on synthetic and real-world networks validate our theoretical findings.

\end{abstract}

\section{Introduction}

Network data, such as the connections found in social networks, often contain sensitive information. Therefore, protecting individual privacy during data analysis is essential. Differential Privacy (DP)~\citep{dwork2014algorithmic} has become the standard method for providing strong privacy guarantees. DP ensures that the inclusion or exclusion of any single user's data in a dataset has only a minimal effect on the results of statistical queries.

In analyzing network and graph data, two primary privacy notions have been introduced~\citep{karwa2011private, karwa2014private}: edge DP, which protects individual relationships (edges) by ensuring that the inclusion or removal of any single edge has limited influence on the output, and node DP, which protects an entire node and all of its incident edges. For clustering and community detection, edge DP is generally more appropriate since relationships are central to defining node labels~\citep{chen2023private, seif2025differentially, suppakitpaisarn2025locally}.

Beyond clustering (which is the focus of this paper), DP has become the dominant framework for private graph analysis, offering rigorous guarantees without assumptions on the adversary. DP algorithms have been developed for a broad range of tasks, including subgraph counting~\citep{khayatian2025differentially, roy2025robust} (stars, triangles, cuts, and dense subgraphs), community detection, personalized PageRank~\citep{wang2013differential, epasto2022differentially}, and synthetic graph generation~\citep{nguyen2016detecting, qin2017generating, imola2021locally, blocki:itcs13, hehir2021consistency, cohen2022scalable, mohamed2022differentially}.

\textbf{Related Work.} Spectral clustering~\citep{ng2001spectral, von2007tutorial} is known to be a computationally efficient algorithm for finding clustering in graph datasets. This algorithm has long been studied as a reliable method for community detection, with consistency guarantees established under various random graph models. Foundational results show that the leading eigenvectors of the adjacency or Laplacian matrices recover latent communities when the graph is sufficiently well-separated~\citep{mcsherry2001spectral, lei2015consistency, abbe2017community}. Extensions have further analyzed sparse and degree-heterogeneous settings~\citep{qin2013regularized, le2017concentration, abbe2020entrywise}. However, despite this rich line of work, there remains no clear understanding of spectral clustering when subject to privacy constraints.

Recently, several efforts have investigated spectral clustering under privacy. Hehir \emph{et al.}~\citep{hehir2021consistency} analyzed graph perturbation using randomized response and showed that accurate recovery is guaranteed only when the average degree scales as $\Theta(\sqrt{n})$, a condition much stronger than those typically required in sparse or heterogeneous graphs.~\citet{suppakitpaisarn2025locally} considered the local model, where users randomize their adjacency lists before aggregation, but their analysis is restricted to this particular mechanism and the special case of two communities. A related line of work studies private principal component analysis, most notably the Analyze Gauss mechanism of~\citet{dwork2014analyze} and the noisy power method of~\citet{hardt2014noisy}; we adapt both to edge-DP graph clustering and use them as baselines.

In contrast, our work develops a unified framework of spectral mechanisms under edge-DP, deriving explicit misclassification error bounds and examining the scaling laws of the proposed private mechanisms. Unlike subspace clustering methods~\citep{ng2001spectral, wang2015differentially}, which partition high-dimensional data points lying near multiple low-dimensional subspaces, our focus is on graph clustering, where the structure is defined by pairwise relationships among nodes. To the best of our knowledge, this is the first systematic study of the fundamental trade-offs among privacy, accuracy, and efficiency in spectral clustering without restrictive assumptions on the underlying graph model.

\textbf{Our Contributions.} We make the following key contributions:
\begin{enumerate}[leftmargin=*, itemsep=2pt, topsep=2pt]

\item \textbf{Matrix shuffling mechanism.} We employ randomized response to perturb the adjacency matrix, and incorporate a novel \emph{matrix shuffling mechanism} based on a uniformly random permutation that preserves the spectral properties of the perturbed matrix. We prove that, while randomized response alone provides only a constant $\varepsilon$ guarantee as the graph grows, the shuffled mechanism achieves $(\varepsilon, \delta)$-edge DP with effective $\varepsilon$ that tends to zero as the number of nodes increases.

\item \textbf{Unified error analysis framework.} We develop a unified analysis based on Davis--Kahan subspace perturbation, Procrustes alignment, and a $k$-means classification-margin bound. This framework yields explicit misclassification rates for any private spectral clustering mechanism that admits an additive-noise decomposition, and applies uniformly to all three mechanisms we study. Applying our framework, we show that the matrix shuffling mechanism achieves an error rate scaling of $\tilde{O}(1/n)$, a clear improvement over two canonical DP-PCA baselines adapted to the edge-DP setting: Analyze Gauss~\citep{dwork2014analyze} and the noisy power method~\citep{hardt2014noisy}, both of which scale as $\tilde{O}(1)$ in $n$.


\item \textbf{Private spectral gap detection.} When the number of communities is unknown, we propose a centered-and-projected eigengap procedure that estimates it directly from the privatized output. The procedure inherits the privacy amplification of shuffling, and we establish a high-probability recovery guarantee under a population-gap separation condition.

\item \textbf{Experimental validation.} We validate our theoretical findings on synthetic stochastic block models and on real-world graphs (Facebook Social Circles and Cora). Across all settings, the matrix shuffling mechanism delivers the strongest privacy--utility trade-off, consistent with our theoretical predictions.

\end{enumerate}

We summarize our theoretical results for private clustering in Table~\ref{tab:mechanisms_summary}, which highlights the fundamental trade-offs between error rate, computational cost, and space complexity. The matrix shuffling mechanism achieves the lowest error rate among the three approaches, but at the expense of higher computational complexity due to the densification induced by the randomized response perturbation. The noisy power method offers a favorable balance between accuracy and efficiency, particularly in the dense-graph regime where the eigengap scales as $\Theta(\log n)$. Analyze Gauss provides the cleanest noise prefactor but, like shuffling, requires $O(n^2)$ storage. Missing proofs and additional discussions are provided in the Appendix.

\section{Problem Statement  \& Preliminaries}
\label{sec:preliminaries_and_problem_statement}

We consider an undirected graph $G=(V,E)$ consisting of $n$ vertices, 
where $V=\{1,\dots,n\}$ is the vertex set and $E \subseteq \{\{i,j\}: i,j \in V,\, i\neq j\}$ 
is the edge set. 
The vertices are partitioned into $k$ disjoint clusters 
$C_{1},\dots,C_{k}$ such that $V=\bigcup_{\ell=1}^{k} C_{\ell}$. 
The adjacency matrix $\mathbf{A}\in\{0,1\}^{n\times n}$ is defined by
\[
A_{ij}=\mathbf{1}\{\{i,j\}\in E\}, \quad A_{ij}=A_{ji}, \quad A_{ii}=0.
\]
The eigenvalues of the adjacency $\mathbf{A}$ are given as $\lambda_{1}(\mathbf{A}) \geq \lambda_{2}(\mathbf{A}) \cdots \geq \lambda_{n}(\mathbf{A}) \geq 0$.

We now introduce the metric that quantifies the accuracy of the recovered partitions.

\begin{definition}[$(\beta, \eta)$-Accurate Recovery] \label{def:error_rate}
Let $\bm{\sigma}^{*} = \{\sigma^{*}_{1}, \sigma^{*}_{2}, \ldots, \sigma^{*}_{n}\}$ denote the ground-truth cluster assignment of $n$ nodes into $k$ (unknown) clusters, where $\sigma^{*}_{i} \in [k]$.  
A recovery algorithm outputs an estimated assignment $\hat{\bm{\sigma}}(G) = \{\hat{\sigma}_{1}, \hat{\sigma}_{2}, \ldots, \hat{\sigma}_{n}\}$ with $\hat{\sigma}_{i} \in [k]$.  

We say that a clustering algorithm $\hat{\bm{\sigma}}$ achieves $(\beta, \eta)$\emph{-accurate recovery} (up to a global permutation of clusters labels) if
\begin{align}\label{eqn:beta_eta_accurate_definition}
\Pr\Bigl( \operatorname{err}\operatorname{rate}\bigl(\hat{\bm{\sigma}}(G), \bm{\sigma}^{*}\bigr) \leq \beta \Bigr) \;\geq\; 1 - \eta,
\end{align}
where the probability is taken over the randomness of the clustering algorithm.  

Here, the error rate is defined by the normalized Hamming distance with optimal label assignment:
\begin{align}
\operatorname{err}\operatorname{rate}\bigl(\hat{\bm{\sigma}}(G), \bm{\sigma}^{*}\bigr)
:=
\frac{1}{n}\cdot
\min_{\pi \in S_{k}}
\operatorname{Ham} \bigl(\pi(\hat{\bm{\sigma}}(G)),\, \bm{\sigma}^{*}\bigr), \nonumber 
\end{align}
where $S_{k}$ is the set of all permutations of $[k]$, and $\pi(\hat{\bm{\sigma}}(G))$ denotes relabeling the estimated clusters according to $\pi$.
\end{definition}

\textbf{Differential Privacy on Graphs.} We focus on the notion of \emph{edge} privacy~\citep{blocki2013differentially}, where two graphs $G$ and $G'$ are said to be neighbors, denoted $G \sim G'$, if they share the same vertex set but differ in exactly one edge. The formal definition is given next.

 \begin{definition} [$(\varepsilon, \delta)$-edge DP] \label{def:edgeDP}  A (randomized) clustering algorithm $\hat{\bm{\sigma}}$ as a function of a graph $G$ satisfies $(\varepsilon, \delta)$-edge DP for some $\varepsilon \in \mathds{R}^{+}$ and $\delta \in [0,1]$, if for all pairs of adjacency matrices $G$ and $G'$ that differ in {\it one} edge, and any measurable subset $\mathcal{S} \subseteq \operatorname{Range}(\hat{\bm{\sigma}})$, we have 
\begin{align*}
    {\operatorname{Pr}(\hat{\bm{\sigma}}(G) \in \mathcal{S}}) \leq e^{\varepsilon} { \operatorname{Pr}(\hat{\bm{\sigma}}(G') \in \mathcal{S}}) + \delta.
\end{align*}
The setting when $\delta = 0$ is referred as $\varepsilon$-edge DP.
\end{definition}

\section{General Framework: Private Spectral Clustering}

\textbf{Spectral clustering.} Spectral clustering assembles the top-$k$ eigenvectors  $\mathbf{U} = [\mathbf{u}_1, \ldots, \mathbf{u}_k] \in \mathbb{R}^{n \times k}$  of the adjacency matrix $\mathbf{A}$ and applies $k$-means to the rows  of~$\mathbf{U}$. Each row $\mathbf{u}_i^\top \in \mathbb{R}^k$ serves as a low-dimensional  embedding of node~$i$; when the eigenvalues satisfy  $\lambda_k(\mathbf{A}) > \lambda_{k+1}(\mathbf{A})$, nodes in the same  community are mapped to nearby points, so that $k$-means on these  embeddings recovers the partition under standard separation  conditions~\citep{von2007tutorial, mcsherry2001spectral}. 
All mechanisms in this paper follow this template---they differ only in how the basis $\mathbf{U}$ is obtained privately. The $k$-means clustering can be seen as post-processing of DP results.

\textbf{Error Analysis.} We next introduce two standard lemmas that will be useful in our analysis of the methods. 
The first lemma, a variant of the classical Davis--Kahan theorem, 
controls the deviation between the true eigenspace and the perturbed eigenspace 
in terms of the spectral gap and the perturbation magnitude.


To analyze the accuracy of the matrix shuffling mechanism and the Analyze Gauss, we represent the privatized output $\widetilde{\mathbf{A}}$ via an additive noise structure. Specifically, we consider the decomposition
\begin{equation}
    \widetilde{\mathbf{A}} = c \mathbf{A} + \mathbf{S} + \mathbf{Z},
\end{equation}
where $c > 0$ is a mechanism-dependent scaling factor, $\mathbf{S}$ is a deterministic bias matrix, and $\mathbf{Z}$ is a zero-mean random noise matrix. In this framework, the perturbation to the eigenspace is  driven by the spectral norm of $\mathbf{Z}$, which contains the randomness required to satisfy edge-DP.

\begin{lemma}[Davis--Kahan~\citep{stewart1998perturbation}]
\label{lem:DK_explicit}
Let $\mathbf{U}=[\mathbf{u}_1,\cdots,\mathbf{u}_k]$ and $\widetilde{\mathbf{U}}$ collect the top-$k$ eigenvectors of $\mathbf{A}$ and $\widetilde{\mathbf{A}}$, respectively. 
If $\Delta_k = \lambda_k(\mathbf{A})-\lambda_{k+1}(\mathbf{A})>0$, then
\[
\operatorname{dist}(\mathbf{U},\widetilde{\mathbf{U}})
\;\le\; 2 \cdot\frac{\|\mathbf{Z}\|_2}{\Delta_k}.
\]
\end{lemma}

The second lemma, a Procrustes alignment bound, shows that after an appropriate orthogonal rotation, 
the Frobenius error between the two embeddings can be controlled by the eigenspace distance.

\begin{lemma}[Procrustes alignment]
\label{lem:procrustes}
There exists an orthonormal matrix $\mathbf{R}\in\mathbb{O}_{k}$ such that
\[
\|\widetilde{\mathbf{U}}-\mathbf{U} \mathbf{R}\|_F
\;\le\;
\sqrt{2k}\, \cdot \operatorname{dist}(\mathbf{U},\widetilde{\mathbf{U}}).
\]
\end{lemma}

The following result bridges the perturbation error with the final clustering performance. 
In particular, it connects the misclassification rate of $k$-means clustering with the classification margin, 
captured by the separation between cluster centers $\Delta_\star$, and the within-cluster radius $r_\star$.

\begin{lemma}[Error Rate Analysis via Classification Margin]
\label{lem:kmeans_margin}
Let $C_1,\ldots,C_k$ be the ground-truth clusters in the clean embedding $\mathbf{U}$, with centers 
$c_r := |C_r|^{-1}\sum_{i\in C_r} \mathbf{u}_i$. 
Define the within-cluster radius and separation as
\begin{equation*}
\begin{aligned}
r_\star  := \max_r \Big(|C_r|^{-1}\sum_{i\in C_r}\|\mathbf{u}_i-c_r\|^2\Big)^{1/2}, \quad 
\Delta_\star  := \min_{r\ne s}\|c_r-c_s\|.
\end{aligned}
\end{equation*}
Then, the misclassification error rate can be bounded as 
\begin{equation*}
\begin{aligned}
\mathrm{err\,rate}(\hat{\bm{\sigma}},\bm{\sigma}^\ast)
 \;\le\;
 \frac{\|\widetilde{\mathbf{U}}- \mathbf{U} \mathbf{R}\|^{2}_F}{n \cdot \bigl(\tfrac{\Delta_\star}{2}-r_\star\bigr)^2}.    
\end{aligned}
\end{equation*}
\end{lemma}
We next consider spectral clustering mechanisms that operate under edge DP. 
In particular, we consider three design approaches:
(i) matrix perturbation and shuffling of the adjacency matrix $\mathbf{A}$;
(ii) noisy power method~\citep{hardt2014noisy}, which integrate carefully calibrated noise into noisy power iterations with $\mathbf{A}$, ensuring privacy at each step while maintaining convergence to informative eigenvectors;
(iii) Analyze Gauss~\citep{dwork2014analyze} which adds Gaussian distributed noise to the elements of $\mathbf{A}$ and computes the eigenvectors as post-processing. 
Using directly the adjacency matrix $\mathbf{A}$ is natural for the chosen methods since they are designed to approximate leading eigenvectors, and choosing $\mathbf{A}$ results in clean sensitivity bounds. 


\section{Matrix Shuffling Mechanism} \label{sec:shuffling}

First, we would like to make an observation that the graph structure is invariant to permutations, so instead of releasing the adjacency matrix $\mathbf{A}$, releasing $\mathbf{P} \mathbf{A} \mathbf{P}^{\top}$ with any permutation matrix $\mathbf{P}$ will give the same information about the graph structure and in particular does not affect the error rate given in Def.~\ref{def:error_rate}. With this insight, we consider releasing $\mathbf{P} \mathbf{A} \mathbf{P}^{\top}$ with a random permutation and show how to use the so-called shuffling amplification results~\citep{feldman2022hiding} to amplify the DP guarantees induced by the graph perturbation mechanism which is defined next. \\

\begin{definition}[Graph Perturbation Mechanism]

Let $\mathbf{A}$ be the adjacency matrix of the graph.  
Under Warner's randomized response (RR) mechanism~\citep{warner1965randomized} with perturbation parameter 
$\mu = 1/(e^{\varepsilon_{0}}+1)$, we generate a privatized adjacency matrix
\begin{equation}
    \widetilde{\mathbf{A}} = \mathbf{A} + \mathbf{E},
\end{equation}
where $\mathbf{E}$ is a symmetric random perturbation matrix.  
For each off-diagonal entry $(i,j)$ with $i<j$, the perturbation is defined as
\begin{equation} \label{eq:E_def}
E_{ij} =
\begin{cases}
    0, & \text{with probability } 1-\mu, \\
    1 - 2A_{ij}, & \text{with probability } \mu,
\end{cases}
\end{equation}
and we enforce $E_{ji} = E_{ij}$ to ensure symmetry, while setting $\widetilde{A}_{ii} = 0$.  
By construction, the entries satisfy
\begin{equation*}
\begin{aligned}
\mathbb{E}[E_{ij}] = \mu \,\bigl(1 - 2A_{ij}\bigr), \quad
\operatorname{Var}(E_{ij}) = \mu(1-\mu)\,\bigl(1 - 2A_{ij}\bigr)^2. \nonumber
\end{aligned}
\end{equation*}
\end{definition}

\noindent \textbf{Privacy Amplification via Matrix Shuffling.} The setting in the existing results is slightly different from ours as they consider a random permutation of $n$ local DP mechanisms, however we are able to get a amplification result for the random permutation of a randomized response perturbed adjacency matrix $\mathbf{A}$. Moreover, in the shuffling amplification results of~\citep{feldman2022hiding}, one considers  $n$ binary data points $x_1, \ldots, x_n$ and applying randomized response to each of them (with $\veps_0$-DP) and then shuffling randomly. The random permutation of these local DP results is $(\veps,\delta)$-DP where
$
\veps = O( n^{-1/2} \sqrt{e^{\veps_0} \log 1/\delta} )
$
and there are explicit upper bounds for $\veps$ that can be used for amplifying the existing randomized response results. 
In particular, the analysis is based on decomposing individual local DP contributions to mixtures of data dependent part and noise, which leads to finding $(\veps,\delta)$-bound for certain 2-dimensional discrete-valued distributions. We next present our privacy amplification result in the following theorem.

\begin{theorem} \label{thm:shuffling_bound}
    
Let $\mathbf{A}$ and $\mathbf{A}'$ be two symmetric adjacency matrices differing in the $(i,j)$\textsuperscript{th} element.
Define the mechanism $\mathcal{M}(\mathbf{A})$ as: 
$$
\mathcal{M}(\mathbf{A}) = \mathbf{P}\widetilde{\mathbf{A}}\mathbf{P}^{\top},
$$
where $\widetilde{\mathbf{A}}$ is the adjacency matrix $\mathbf{A}$ perturbed using $\veps_0$-DP randomized response (i.e., strictly upper triangular elements of $\mathbf{A}$ are perturbed with $\veps_0$-DP randomized response and mirrored to the lower triangular part) and $\mathbf{P}$ is a permutation matrix corresponding to a randomly drawn permutation $\pi$ of $[n]$, i.e.,  
$
\mathbf{P} = \begin{bmatrix}
    \mathbf{e}_{\pi(1)} & \ldots & \mathbf{e}_{\pi(n)}
\end{bmatrix}^{\top},
$
where $\mathbf{e}_j$, $j \in [n]$, denotes the $j$\textsuperscript{th} unit vector.
Then, for any $\delta \in (0,1]$ and $\veps_0 \leq \log \left( \frac{n}{8 \log(2 / \delta)} - 1 \right)$, the mechanism $\mathcal{M}(A)$ is $(\veps,\delta)$-DP for 
\begin{equation} \label{eq:ineq_hockey_stick_shuffling}
\veps \leq \log \left( 1 + (e^{\veps_0} - 1) \left( \frac{4 \sqrt{2\log(4/\delta)}}{\sqrt{(e^{\veps_0}+1)n}} + \frac{4}{n} \right) \right).
\end{equation}
\end{theorem}

We next analyze the accuracy of our proposed mechanism. To this end, we first focus on the error induced by the randomized response perturbation and establish a set of intermediate lemmas that will serve as building blocks for proving the final guarantee on the misclassification error rate.

\begin{lemma}[Matrix Representation for $\widetilde{\mathbf{A}}$]
\label{lem:decomp_explicit}
Let $\mathbf{J}=\mathbf{1}\mathbf{1}^\top$. Under symmetric RR with edge flip probability $\mu\in[0,\tfrac12)$, the perturbed adjacency matrix can be represented as
\[
\widetilde{\mathbf{A}}
\;=\;
c\,\mathbf{A} \;+\; \mu(\mathbf{J}-\mathbf{I}) \;+\; \mathbf{Z},
\quad
c := 1 - 2\mu,
\]
where $\mathbf{Z} := \widetilde{\mathbf{A}} - \mathbb{E}[\widetilde{\mathbf{A}}\mid \mathbf{A}]$ has independent mean-zero entries, bounded by $1$ in absolute value, and variance at most $v_{\max} := \mu(1-\mu)$.
\end{lemma}
We next present an upper bound on the spectral norm of the matrix $\mathbf{Z}$. 
\begin{lemma}[Bounding the Spectral Norm of $\mathbf{Z}$]
\label{lem:bern_Z}
Let $\mathbf{Z}$ be defined as above. Then, for any $\eta\in(0,1)$, with probability at least $1-\eta$, we have
\[
\|\mathbf{Z}\|_2 
\;\le\; 
\sqrt{\,2(n-1)v_{\max}\,\log\frac{2n}{\eta}\,}
\;+\;\frac{1}{3}\,\log\frac{2n}{\eta}.
\]
\end{lemma}

We next give the error rate for the plain RR perturbed mechanism. 
Compared to a vanilla $k$-means cost analysis, the classification–margin–based approach provides a more fine-grained characterization of performance, as it explicitly captures how the privacy budget, the number of communities, and the spectral gap jointly influence the misclassification error. 
\begin{theorem}[Spectral $k$-means under RR perturbation of $\mathbf{A}$]
\label{thm:main_explicit}
Let $\mathbf{A}$ be the adjacency matrix of the clean graph, with eigengap 
$\Delta_k=\lambda_{k}(\mathbf{A})-\lambda_{k+1}(\mathbf{A})>0$. 
Apply RR with flip probability $\mu$, and let $\widetilde{\mathbf{A}}$ be the perturbed adjacency matrix. 
Let $\hat{\bm{\sigma}}$ be the (global optimum) $k$-means clustering of the top-$k$ eigenvectors of $\widetilde{\mathbf{A}}$. 
Then, with probability at least $1-\eta$,
\begin{equation*}
\begin{aligned}
\mathrm{err\,rate}(\hat{\bm{\sigma}},\bm{\sigma}^\ast)  \;\le\;
\frac{\Bigl[
\displaystyle \frac{2\sqrt{k}}{\Delta_k}\,\sqrt{\,2(n-1)v_{\max}\,\log\frac{2n}{\eta}\,}
\;+\;\frac{2\sqrt{k}}{3\Delta_k}\,\log\frac{2n}{\eta}
\Bigr]^2}
{n \, \bigl(\tfrac{\Delta_\star}{2}-r_\star\bigr)^2},
\end{aligned}
\end{equation*}
where $v_{\max}=\mu(1-\mu)$. 
\end{theorem}
Combining the error rate given by Theorem~\ref{thm:main_explicit} and the shuffling amplification result of Theorem~\ref{thm:shuffling_bound}, Lemma~\ref{lem:asymptotic_shuffling} in the Appendix shows that the matrix shuffling mechanism achieves an asymptotic error rate of $\tilde O(1/n)$. This is in contrast to the baseline mechanisms, whose error rates remain $\tilde{O}(1)$ in $n$, owing to the privacy amplification effect, as shown in Table~\ref{tab:mechanisms_summary}.

\section{Baseline Mechanisms}

\subsection{Noisy Power Method}


We consider another approach to perform spectral clustering, the noisy power method \citep{hardt2014noisy}.
The noisy power method can be used to estimate the leading $k$ eigenvectors of $\mathbf{A}$ (denoted by $\mathbf{U} \in \mathbb{R}^{n \times k}$) with DP guarantees. Given an initial random orthogonal matrix $\mathbf{X}_0 \in \mathbb{R}^{n \times k}$, the method repeats $\mathbf{X}_i \leftarrow \mathrm{QR}(\mathbf{A}\,\mathbf{X}_{i-1} + \mathbf{Z}_i)$ for $i=1,\dots,N$, where $\mathbf{Z}_i \in \mathbb{R}^{n\times k}$ has i.i.d.\ $\mathcal{N}(0,\,C^2\bar{\sigma}^2)$ entries, and returns $\mathbf{X}_N$ as the estimate of $\mathbf{U}$.
The sensitivity constant $C$ has to upper bound the Frobenius norm sensitivity of $\mathbf{A}\mathbf{X}_{i-1}$ w.r.t. changes of edges. Notice that since for an orthogonal $\mathbf{X}$, we have that $\|(\mathbf{A}-\mathbf{A}')\mathbf{X}\|_F \leq \|\mathbf{A}-\mathbf{A}'\|_F \cdot \|\mathbf{X}\|_2 = \sqrt{2}$ for any $\mathbf{A}$ and $\mathbf{A}'$ that differ by one edge. Thus we can set $C=\sqrt{2}$. 
The procedure is depicted in the pseudocode of Alg.~\ref{alg:dp_noisy_power}.


From~\citep{hardt2014noisy} we directly get the following error bound where the error
$\| (\mathbf{I} - \mathbf{X}_i \mathbf{X}_i^T) \mathbf{U} \|$
depends inversely on $\varepsilon$ and inversely on the eigengap $\lambda_{k}(\mathbf{A}) - \lambda_{k+1}(\mathbf{A})$. 

\begin{lemma} \label{lem:ref_bound}
Let $\Delta_k=\lambda_{k}(\mathbf{A})-\lambda_{k+1}(\mathbf{A})>0$
denote the $k$th eigengap of the matrix $\mathbf{A}$. If we choose $\bar{\sigma} = \varepsilon^{-1} \sqrt{4 N \log(1/\delta)}$, then the noisy power method with $N$ iterations satisfies $(\varepsilon, \delta)$-edge DP.  Moreover, after $ N = O\left( \frac{\lambda_k}{\Delta_k} \log n  \right) $ iterations we have with probability at least $1-\eta$ that
\begin{equation*} 
\begin{aligned}
     \left\|  (\mathbf{I} - \mathbf{X}_N \mathbf{X}_N^T) \mathbf{U} \right\|_2 =  
     O\bigg( \frac{ \bar{\sigma} \, \left( \sqrt{n} + \sqrt{2 \log\left( {2N}/{\eta} \right)} \right)}{\Delta_k}  \cdot \frac{\sqrt{k+1}}{\sqrt{k+1} - \sqrt{k}} \bigg).
\end{aligned}
\end{equation*}
\end{lemma}

Combining Lemmas~\ref{lem:kmeans_margin} and~\ref{lem:ref_bound} and the Procustres alignment bound of Lemma~\ref{lem:projection_to_U}, we directly get the following bound for the classification error rate of the noisy power method.

\begin{theorem}[Spectral $k$-means under Noisy Power Method]
\label{thm:main_explicit_power_method}
Let $\mathbf{A}$ be the adjacency matrix of a graph $G$, with eigengap $\Delta_k=\lambda_{k}(\mathbf{A})-\lambda_{k+1}(\mathbf{A})>0$. 
Let $\mathbf{R}\in\mathbb{O}_k$ be the Procrustes alignment of $\widetilde{\mathbf{U}}$ to $\mathbf{U}_Q$.
Define clean embedding geometry on $\mathbf{U}$: cluster centers
$\mathbf{c}_r:=|C_r|^{-1}\sum_{i\in C_r}(\mathbf{U})_{i,:}$, separation
$\Delta_\star:=\min_{r\ne s}\|\mathbf{c}_r-\mathbf{c}_s\|$, and within–cluster radius
$r_\star:=\max_r\big(|C_r|^{-1}\sum_{i\in C_r}\|(\mathbf{U})_{i,:}-\mathbf{c}_r\|^2\big)^{1/2}$.
Choosing $\bm{\sigma}$ and $N$ as in the statement of Lemma~\ref{lem:ref_bound}, we have with probability at least $1-\eta$ that
\begin{equation*}
    \begin{aligned}
\mathrm{err\,rate}(\hat{\bm{\sigma}},\bm{\sigma}^\ast)
=
 O\!\left( \frac{ k^3 \, \bar{\sigma}^2 \, \left( \sqrt{n} + \sqrt{2 \log\left( {2N}/{\eta} \right)} \right)^2}{n \, \Delta_k^2\,(\tfrac{\Delta_\star}{2}-r_\star)^2}
\right).        
    \end{aligned}
\end{equation*}
\end{theorem}
From the error rate expression above, we observe that the noisy power method initially appears to scale polynomially with the number of communities, as $k^{3}$. At first glance, this suggests a worse dependency compared to the other two mechanisms. However, as shown in Appendix~\ref{app:power_rate}, after substituting the parameters of the noisy power method and simplifying, the error rate can be expressed as
$\tilde{O}\!\left({{k^{3}}/{\varepsilon^2 \Delta_k^{3}}}\right)$, 
which also depends on the eigengap $\Delta_{k}$, itself a function of $n$. For instance, in the dense regime of the SBM, we have $\Delta_{k} = \Theta(\log n)$. Consequently, for sufficiently large $n$, the scaling improves and the noisy power method can achieve competitive performance.
We comment on the reason why shuffling cannot be directly applied to the noisy power method in Appendix~\ref{sec:difficulty}.

\subsection{Analyze Gauss} \label{subsec:analyze_gauss}

The Analyze Gauss mechanism, introduced by~\citet{dwork2014analyze} for row-level DP-PCA, releases a noisy version of the covariance matrix and extracts eigenvectors as post-processing. We adapt this idea to the edge-DP graph clustering setting by adding symmetric Gaussian noise directly to the adjacency matrix $\mathbf{A}$.

\textbf{Mechanism.} Release $\widetilde{\mathbf{A}} = \mathbf{A} + \mathbf{E}$, where $\mathbf{E} \in \mathbb{R}^{n \times n}$ is a symmetric matrix whose upper-triangular entries (including the diagonal) are drawn i.i.d.\ from $\mathcal{N}(0, \bar{\sigma}^2)$ and mirrored to the lower triangle. The top-$k$ eigenvectors of $\widetilde{\mathbf{A}}$ are then computed and passed to $k$-means. Since eigenvector extraction and clustering are deterministic post-processing of the private release $\widetilde{\mathbf{A}}$, the entire pipeline inherits the privacy guarantee of the noise addition step.

\textbf{Privacy.} When two symmetric adjacency matrices $\mathbf{A}$ and $\mathbf{A}'$ differ in a single edge $(i,j)$, 
the Frobenius-norm sensitivity is $\|\mathbf{A} - \mathbf{A}'\|_F = \sqrt{2}$. By standard arguments for the Gaussian mechanism~\citep{dwork2006calibrating}, we have $(\varepsilon,\delta)$-edge DP guarantees by setting
$\bar{\sigma} \;=\; \sqrt{2}\,\sqrt{2\log(1.25/\delta)} / \varepsilon$.

\textbf{Error analysis.} Since $\mathbf{E}$ is a symmetric Gaussian ensemble with entry variance $\bar{\sigma}^2$, standard random matrix theory gives $\|\mathbf{E}\|_2 = O(\bar{\sigma}\sqrt{n})$ with h.p.~\citep[Sec.\;7.4][]{vershynin2018high}. Applying Davis--Kahan (Lemma~\ref{lem:DK_explicit}) and the Procrustes and margin bounds (Lemmas~\ref{lem:procrustes}--\ref{lem:kmeans_margin}) gives the following.

\begin{theorem}[Spectral $k$-means under Analyze Gauss]
\label{thm:analyze_gauss}
Let $\mathbf{A}$ be the adjacency matrix with eigengap 
$\Delta_k = \lambda_k(\mathbf{A}) - \lambda_{k+1}(\mathbf{A}) > 0$. 
Let $\widetilde{\mathbf{A}} = \mathbf{A} + \mathbf{E}$ with 
$\bar{\sigma} = \Theta(\sqrt{\log(1/\delta)}/\varepsilon)$ as above. 
Then,
\[
\mathrm{err\,rate}(\hat{\bm{\sigma}}, \bm{\sigma}^\ast)
\;=\;
\tilde{O}\!\left(\frac{k}
{\varepsilon^2\,\Delta_k^2\,(\Delta_\star/2 - r_\star)^2}\right).
\]
\end{theorem}



\subsection{Comparison of All Three Mechanisms}

Table~\ref{tab:mechanisms_summary} summarizes the three mechanisms via the asymptotic rates. All error rates share the factor $k/\Delta_k^2$ from the shared Davis--Kahan and margin analysis; they differ only in the noise prefactor. The shuffling mechanism achieves a $n^{-1}/(e^\varepsilon - 1)^2$ prefactor thanks to the privacy amplification of Theorem~\ref{thm:shuffling_bound}, a clear improvement over both baselines, which are $\tilde O(1)$ in $n$. Analyze Gauss has the cleaner $1/\varepsilon^2$ prefactor, while the noisy power method pays an additional $k^2/\Delta_k$ due to iterative composition. 

The trade-off is computational: shuffling and Analyze Gauss both require $O(n^2)$ storage (dense perturbed matrix), whereas the power method needs only $O(n)$ for sparse graphs. We also note that the theoretical ordering of the two baselines can reverse in practice, since the power method confines noise to $n \times k$ matrices per iteration and each multiplication by $\mathbf{A}$ amplifies the leading eigenspace relative to the noise---effects not captured by the worst-case bounds.

\begin{table*}[h!]
\centering
\caption{Comparison of $(\varepsilon, \delta)$-edge DP spectral clustering mechanisms.  
Baselines: Noisy Power Method~\citep{hardt2014noisy}, 
Analyze Gauss~\citep{dwork2014analyze}.
$E_{\text{tot}}$: total edges; 
${E}_{\text{eff}}= (1-\mu)E_{\text{tot}}+\mu(\binom{n}{2}-E_{\text{tot}})$;
$\Delta_k$: $k$th eigengap of $\mathbf{A}$;
$\tilde{O}(\cdot)$ hides $\log n$ factors.}
\label{tab:mechanisms_summary}
\begin{tabular}{@{}lccc@{}}
\toprule
& \textbf{Matrix Shuffling (ours)} 
& \textbf{Noisy Power Method} 
& \textbf{Analyze Gauss} \\
\midrule
Error  
& $\tilde{O}\!\Bigl(\dfrac{1}{n\,(e^{\varepsilon}-1)^2}\cdot\dfrac{k}{\Delta_k^{2}}\Bigr)$ 
& $\tilde{O}\!\Bigl(\dfrac{k^{2}}{\varepsilon^2 \Delta_k}\cdot\dfrac{k}{\Delta_k^{2}}\Bigr)$ 
& $\tilde{O}\!\Bigl(\dfrac{1}{\varepsilon^2}\cdot\dfrac{k}{\Delta_k^{2}}\Bigr)$ \\[8pt]
Time  
& $\tilde{O}\!\bigl({E}_{\text{eff}} + \sqrt{{E}_{\text{eff}} \log n} + nk^{2}\bigr)$
& $\tilde{O}\bigl((E_{\text{tot}} + nk)N\bigr)$ 
& $O(n^2 k)$ \\[4pt]
Space 
& $O\!\bigl(E_{\text{eff}} + \sqrt{{E}_{\text{eff}} \log n}\bigr)$
& Dense: $O(n\log n)$; Sparse: $O(n)$ 
& $O(n^2)$ \\
\bottomrule
\end{tabular}
\end{table*}

\section{Private Spectral Gap Detection}

In practice, the number of clusters $k$ may be unknown and must be selected from the privatized output. We consider the spectral gap detection under the matrix shuffling mechanism $\mathcal{M}(\mathbf{A}) = \mathbf{P} \widetilde{\mathbf{A}} \mathbf{P}^{\top}$. A key observation is that for any permutation matrix $\mathbf{P}$, the operation $\mathbf{P} \mathbf{A} \mathbf{P}^{\top}$ is a similarity transformation. Since similarity transformations preserve the spectrum, the eigenvalues of the shuffled adjacency matrix are identical to those of the unshuffled perturbed matrix $\widetilde{\mathbf{A}}$. Consequently, we can utilize the privacy amplification results of Section~\ref{sec:shuffling} to achieve a smaller noise level $\mu$ for a given privacy budget $(\varepsilon, \delta)$ while keeping the spectral gap signal untouched.

\textbf{Centered and projected matrix.}
Let $\widetilde{\mathbf{A}}_{\text{shuf}} = \mathcal{M}(\mathbf{A})$. Under the Randomized Response (RR) perturbation with parameter $\mu=1/(e^{\varepsilon_0}+1)$, and noting that $\mathbf{P}(\mathbf{J}-\mathbf{I})\mathbf{P}^\top = \mathbf{J}-\mathbf{I}$ for any permutation matrix $\mathbf{P}$, the shuffled matrix can be represented as:
\begin{equation*}
\widetilde{\mathbf A}_{\text{shuf}} \;=\; c(\mathbf{P}\mathbf{A}\mathbf{P}^\top) + \mu(\mathbf{J}-\mathbf{I}) + \mathbf{Z}_{\text{shuf}},
\qquad c=1-2\mu,
\end{equation*}
where $\mathbf Z_{\text{shuf}}$ is the shuffled zero-mean noise. We first remove the deterministic shift by centering the shuffled output:
\begin{equation*}
\mathbf B \;:=\; \widetilde{\mathbf A}_{\text{shuf}} - \mu(\mathbf J-\mathbf I) \;=\; c(\mathbf{P}\mathbf{A}\mathbf{P}^\top) + \mathbf{Z}_{\text{shuf}} .
\end{equation*}
Next, we remove the trivial leading direction by projecting onto the space orthogonal to the all-ones vector, $\mathbf 1^\perp$. Let $\mathbf H \;:=\; \mathbf I - \frac{1}{n}\mathbf 1\mathbf 1^\top,$ where $\mathbf B_\perp \;:=\; \mathbf H \mathbf B \mathbf H.$
The projection eliminates the domination of the leading eigenpair corresponding to global graph density effects, which can otherwise obscure the informative eigengap reflecting community structure.

\textbf{Gap detection rule.}
Let $\hat\lambda_1\ge \hat\lambda_2\ge \cdots \ge \hat\lambda_r$ denote the top-$r$ eigenvalues of $\mathbf B_\perp$ for some search cap $r\ge 2$, and define the empirical gaps $\hat g_i := \hat\lambda_i-\hat\lambda_{i+1}$ for $i=1,\dots,r-1$. We select the number of community-based dimensions as:
\begin{equation*}
\hat k \;\in\; \arg\max_{1\le i\le r-1} \hat g_i, \qquad\text{and set}\qquad \widehat{k}_{\mathrm{clusters}}:=\hat k+1.
\end{equation*}
The offset by $+1$ reflects that, under balanced stochastic block models (SBMs), the informative community signal for the projected matrix $\mathbf H \mathbf B \mathbf H$ is contained in the first $k-1$ eigenvalues.

Let $\mathbf M := \mathbb E[\mathbf B_\perp \mid \mathbf A] = \mathbf H (c\,\mathbf P\mathbf A\mathbf P^\top)\mathbf H$ and denote its ordered eigenvalues by $\lambda_1(\mathbf M)\ge \lambda_2(\mathbf M)\ge\cdots$. Since the spectrum is invariant to the permutation $\mathbf{P}$, the population gaps $g_i^\star := \lambda_i(\mathbf M)-\lambda_{i+1}(\mathbf M)$ remain well-defined and match the original gaps scaled by $c$.

\begin{theorem}
\label{thm:gap_detection}
Fix $\eta\in(0,1)$ and let $t(\eta) := \sqrt{2(n-1)v_{\max}\log(2n/\eta)} + \frac{1}{3}\log(2n/\eta)$. W.p. at least $1-\eta$, we have $\max_{i} |\hat\lambda_i-\lambda_i(\mathbf M)| \le t(\eta)$ and $\max_{i} |\hat g_i-g_i^\star| \le 2t(\eta)$. Moreover, if the true maximum eigengap satisfies $g_{k^\star}^\star - \max_{i\neq k^\star} g_i^\star > 4t(\eta)$, then $\hat k = k^\star$ with probability at least $1-\eta$.
\end{theorem}

We give an experimental illustration of the DP community size detection in Appendix~\ref{sec:exp_init}. 

\section{Experimental Results}

In this section, we demonstrate the efficacy of the proposed DP algorithms, namely graph perturbation combined with shuffling, matrix projection, and the noisy power method, on both synthetic and real-world graph datasets by examining their privacy–utility trade-offs and computational runtime. The considered datasets are described in Table~\ref{tab:datasets}. To illustrate the difficulty of each clustering task, Table~\ref{tab:datasets} also reports the normalized eigengap $\big(\lambda_k(\mathbf{A}) - \lambda_{k+1}(\mathbf{A})\big)/\lambda_1(\mathbf{A})$, where $k$ is the number of clusters and $\lambda_i(\mathbf{A})$ denotes the $i$th largest eigenvalue of the adjacency matrix $\mathbf{A}$.
For a given value of $\varepsilon$, for the matrix perturbation method combined with the shuffling amplification the randomization parameter $\mu$ is adjusted using the privacy bound~\ref{cor:shuffling_bound} whereas for the projection method and noisy power method the noise parameter $\bar{\sigma}$ is adjusted using  Lemma~\ref{lem:ref_bound}, respectively.
In each experiment, we study the error rate given in Def.~\ref{def:error_rate} for a range of $\varepsilon$-values distributed on a logarithmically equidistant grid, when $\delta= n^{-2}$. All results are averaged over 100 independent runs.
For all datasets, the noisy power method with number of iterations $N=5$. 

\begin{table*}[h!]
\centering
\caption{Summary of graph datasets used in our experiments.}
\label{tab:datasets}
\begin{tabular}{lcccc}
\toprule
\textbf{Graph} & \textbf{Nodes ($n$)} & \textbf{Edges ($E$)} & \textbf{$k$ (Clusters)} \\
\midrule
SBM ($p=0.5$, $q=0.1$) & 600 & 83636 & 3 \\ 
SBM ($p=0.4$, $q=0.15$) & 2000 & 698486 & 10 \\ 
Facebook Social Circles (\verb|egoIdx=4|) & 484 & 42302 & 2 \\
Facebook Social Circles (\verb|egoIdx=3|) & 552 & 22052 & 4 \\
Cora Citation Network  & 2708 & 5429 & 7 \\ 
\bottomrule
\end{tabular}
\end{table*}

\subsection{Stochastic Block Model}

We first consider an undirected stochastic block model (SBM) with $k=3$ communities, each of size $200$. Edges are sampled independently with probability $p=0.5$ between nodes in the same community and $q=0.1$ between nodes in  different communities.
Figure~\ref{fig:sbm3} shows the error rate vs.\ the privacy leakage 
$\varepsilon$  for the three different methods.  We see that, due to the shuffling amplification, the graph perturbation  method gives the best privacy-utility trade-off. 
We then consider a more challenging setting: $k=10$ communities, each of size $200$, with intra-community edge probability $p=0.4$ and cross-community probability $q=0.15$. In this more difficult case the performance of the noisy power method and the matrix projection method drop significantly while the graph perturbation method shows strong performance (Figure~\ref{fig:sbm10}).

\subsection{Facebook Social Circles}

We next construct a multi-class node classification problem from the Facebook Social Circles dataset~\citep{mcauley2012learning}. Each ego-network in this dataset comes with manually annotated “circles” (friendship groups) that provide natural ground-truth communities. To derive classification tasks, we select a specific ego-network and its largest circle, then resolve overlaps using the “drop” policy, which retains only nodes belonging to exactly one selected circle and discards nodes with multi-membership.
We use the ego-networks with indices 3 and 4, selecting the four and two largest circles, respectively. 
As shown by Figures~\ref{fig:fb1} and~\ref{fig:fb2}, the graph perturbation method shows the strongest performance.



\begin{figure*}[h!]
    \centering 
    \subfloat[SBM, $k=3$, $p=0.5$, $q=0.1$]{
        \includegraphics[width=0.46\textwidth]{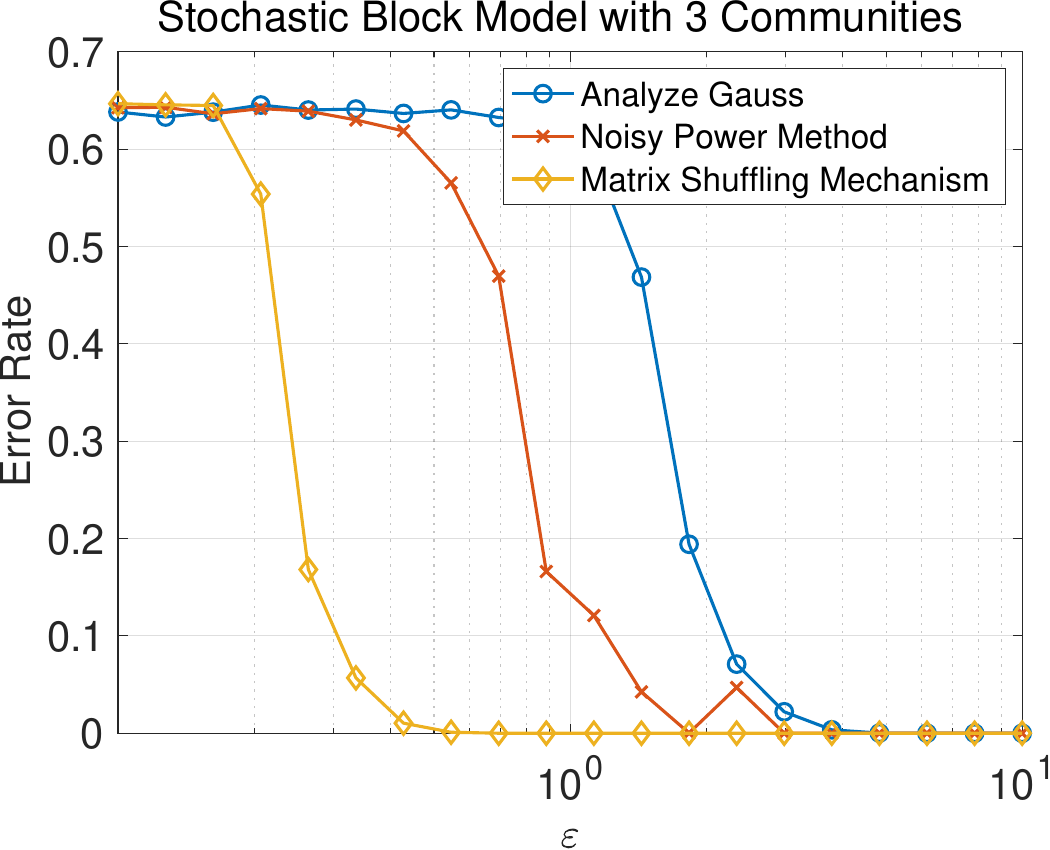}
        \label{fig:sbm3}
    }\hfill 
    \subfloat[SBM, $k=10$,  $p=0.4$, $q=0.15$]{
        \includegraphics[width=0.46\textwidth]{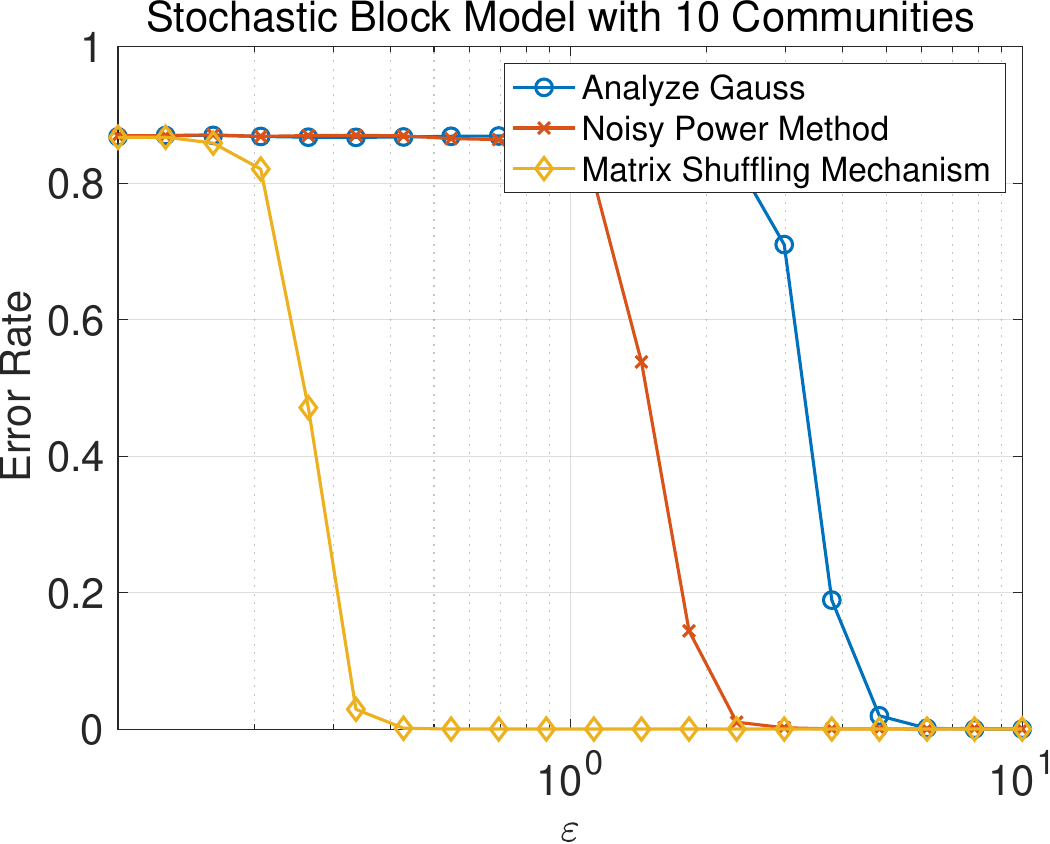}
        \label{fig:sbm10}
    }\hfill \vspace{3mm}
    \subfloat[Facebook, $k=2$]{
        \includegraphics[width=0.46\textwidth]{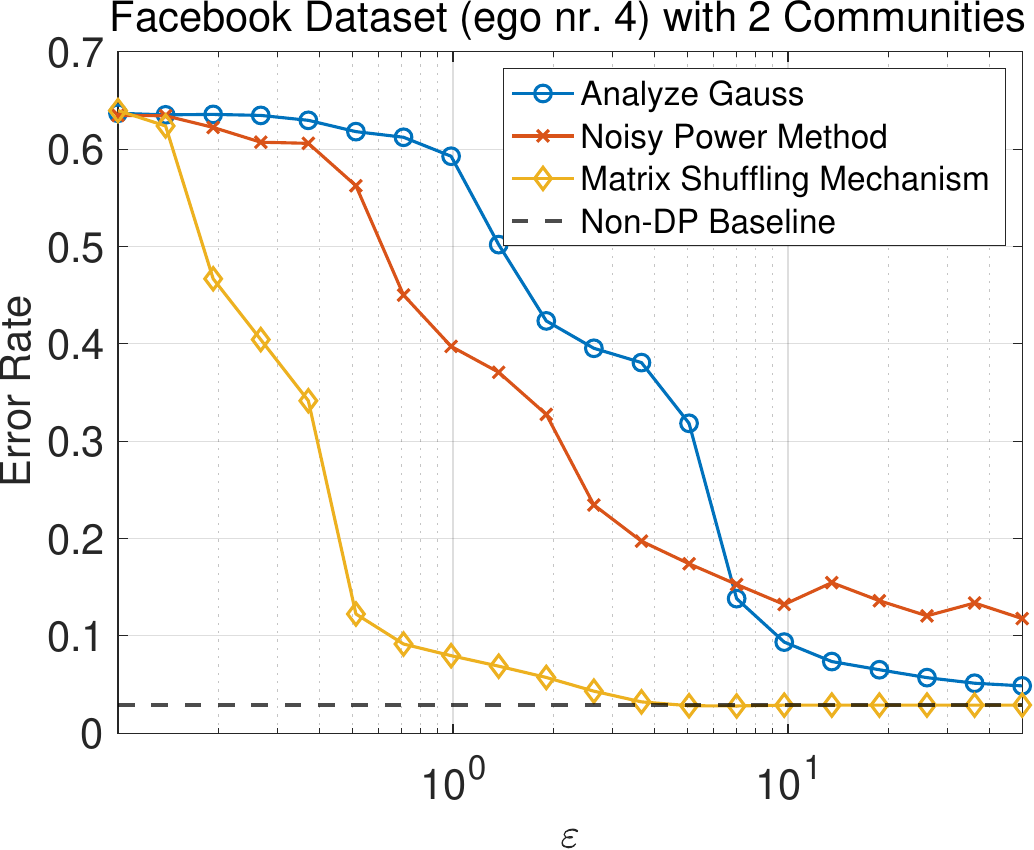}
        \label{fig:fb1}
    }\hfill
    \subfloat[Facebook, $k=4$]{
        \includegraphics[width=0.47\textwidth]{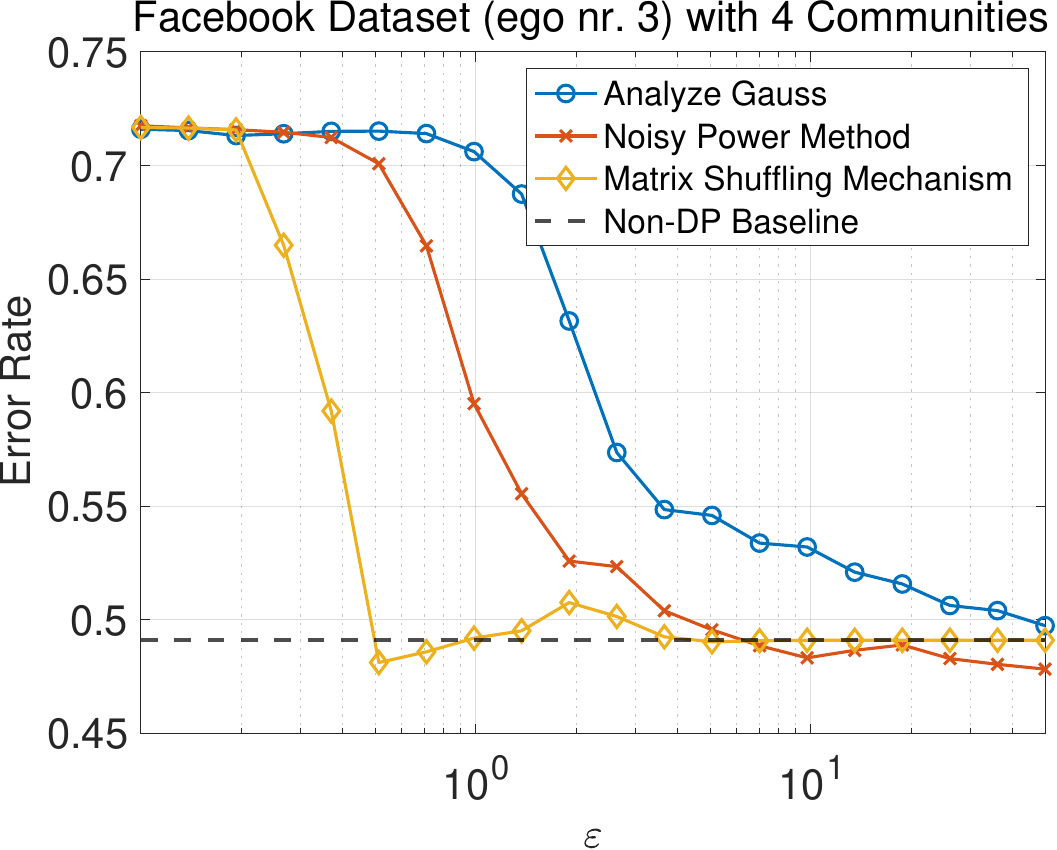}
        \label{fig:fb2}
    }
    \caption{Synopsis of results for the three different mechanisms: error rate vs.\ $\varepsilon$.}
    \label{fig:all-results}
\end{figure*}

\section{Conclusion and Future Work}
\label{sec:conclusion}

In this paper, we studied the problem of spectral clustering for graphs under the notion of edge DP. We analyzed three different classes of privacy-preserving mechanisms, namely, graph perturbation, projected Gaussian mechanism, and edge DP noisy power method. We showed fundamental trade-offs between privacy, accuracy, and efficiency (in terms of computational and space complexity), and established theoretical guarantees on clustering performance. We also presented simulation results on both synthetic and real-world graphs to validate our theoretical findings. Our analysis does have some limitations: the guarantees are developed under edge DP rather than the stronger node DP notion, the matrix shuffling mechanism densifies the adjacency matrix and thus limits scalability to very large sparse graphs, and our error bounds depend on graph-dependent quantities such as the eigengap and cluster geometry. An important direction for future work is to extend these results to \emph{attributed graphs}, where node labels are multidimensional and introduce additional sensitivity challenges, along with broader empirical validation and an extension to the node DP setting.

\bibliographystyle{plainnat}
\bibliography{myreferences}

\newpage

\appendix

\section{Missing Details \& Preliminaries}

The supplementary materials may contain detailed proofs of the results that are missing in the main paper.

\noindent \textbf{Notation.} We use bold uppercase letters to denote matrices (e.g., $\mathbf{A}$) and bold lowercase letters for vectors (e.g., $\mathbf{a}$). The notation $\operatorname{Bern}(p)$ denotes a Bernoulli random variable with success probability $p$. For asymptotic analysis, we write $f(n) = o(g(n))$ if $\lim_{n \to \infty} f(n)/g(n) = 0$. Similarly, $f(n) = O(g(n))$ indicates the existence of a constant $C > 0$ such that $|f(n)/g(n)| \leq C$ for all $n$, and $f(n) = \Omega(g(n))$ indicates the existence of a constant $c > 0$ such that $|f(n)/g(n)| \geq c$ for all $n$. 
For two column-orthonormal matrices $\mathbf{U},\mathbf{U}'\in\mathbb{R}^{n\times k}$, 
the \emph{projection distance} is
\[
\operatorname{dist}(\mathbf{U},\mathbf{U}')
:= \|\mathbf{U}\mathbf{U}^{\top} - \mathbf{U}'\mathbf{U}'^{\top}\|_{2}
= \|\sin\Theta(\mathbf{U},\mathbf{U}')\|_{2},
\]
where $\Theta(\mathbf{U},\mathbf{U}')$ are the principal angles between the two subspaces. 
The Frobenius version $\|\sin\Theta(\mathbf{U},\mathbf{U}')\|_{F}$ is also used when averaging errors. 
Matrix norms $\|\cdot\|_{2}$ and $\|\cdot\|_{F}$ denote spectral and Frobenius norms, respectively.

\subsection{$(\veps,\delta)$-DP Guarantees via the Hockey-Stick Divergence}

The $(\veps,\delta)$-DP guarantees can alternatively be described using the hockey-stick divergence. For $\alpha>0$, the hockey-stick divergence $H_{\alpha}$ from a distribution $P$ to a distribution $Q$ is defined as
$$
H_\alpha(P||Q) = \int \big[P(t) - \alpha \cdot Q(t)\big]_{+} \, d t.
$$
where for $t \in \mathbb{R}$, $[t]_+ = \max\{0,t\}$. 
Tight $(\veps,\delta)$-bounds using the hockey-stick divergence can be characterized as follows~\citep[see Lemma 5,][]{zhu2022optimal}.
\begin{lemma} \label{lem:zhu_lemma5}
For a given $\veps\geq0$, 
tight $\delta(\veps)$ is given by the expression
$$
\delta(\veps) = \max_{G \sim G'} H_{e^\veps}\big(\hat{\bm{\sigma}}(G')||\hat{\bm{\sigma}}(G)\big),
$$ 
where $G \sim G'$ denotes that $G$ and $G'$ differ in one edge.
\end{lemma}
By Lemma~\ref{lem:zhu_lemma5}, for obtaining the edge-DP guarantees,  it is sufficient to bound the hockey-stick divergence for all neighboring graphs.


\subsection{Privacy Threat Model} 

In the context of graph clustering, our threat model is based on edge DP in a centralized setting, where a curator has access to the entire graph, including all nodes and their connections. The primary sensitivity lies in the relationships between nodes rather than in the nodes themselves. This sensitivity arises because exposing how individuals or entities are connected can reveal undisclosed or sensitive information. We assume an adversary may already possess partial knowledge of the graph’s structure and attempt to infer additional undisclosed relationships. Our approach therefore ensures that clustering can be performed without compromising the privacy of individual relationships, while still enabling meaningful analysis of the overall graph structure.

\subsection{Pseudocodes} 

Algorithms~\ref{alg:dp_noisy_power} and~\ref{alg:proj_eigengap} give pseudocodes for the DP noisy power method and for the private spectral gap detection, respectively.

\begin{algorithm}[h!]
\caption{DP Noisy Power Method}
\label{alg:dp_noisy_power}
\begin{algorithmic}[1]
\STATE \textbf{Input}: Adjacency matrix $\mathbf{A}\in\mathbb{R}^{n\times n}$, target rank $k$, iterations $N$, sensitivity $C$, noise scale $\bar{\sigma}$, initialization $\mathbf{X}_0\in\mathbb{R}^{n\times k}$.
\STATE \textbf{Output}: Estimate of the leading $k$ eigenvectors $\mathbf{U}\approx \mathbf{X}_N$.
\FOR{$i = 1,2,\dots,N$}
    \STATE Sample $\mathbf{Z}_i \in \mathbb{R}^{n\times k}$ with entries i.i.d.\ $\mathcal{N}(0,\,C^2 \bar{\sigma}^2)$
    \STATE $\mathbf{Y}_i \leftarrow \mathbf{A}\,\mathbf{X}_{i-1} + \mathbf{Z}_i$
    \STATE $(\mathbf{X}_i,\mathbf{R}_i) \leftarrow \mathrm{QR}(\mathbf{Y}_i)$ \quad (reduced QR decomposition; $\mathbf{X}_i^\top\mathbf{X}_i = \mathbf{I}_k$)
\ENDFOR
\RETURN $\mathbf{X}_N$
\end{algorithmic}
\end{algorithm}

\begin{algorithm}[t]
\caption{Private Spectral Gap Detection}
\label{alg:proj_eigengap}
\begin{algorithmic}[1]
\STATE \textbf{Input:} Privatized adjacency $\widetilde{\mathbf A}$, RR parameter $\mu$, search cap $r$.
\STATE \textbf{Center:} $\mathbf B \leftarrow \widetilde{\mathbf A} - \mu(\mathbf J-\mathbf I)$.
\STATE \textbf{Project:} $\mathbf B_\perp \leftarrow \mathbf H \mathbf B \mathbf H$, where $\mathbf H=\mathbf I-\frac{1}{n}\mathbf 1\mathbf 1^\top$.
\STATE Compute top-$r$ eigenvalues $\hat\lambda_1\ge\cdots\ge\hat\lambda_r$ of $\mathbf B_\perp$.
\STATE Compute gaps $\hat g_i \leftarrow \hat\lambda_i-\hat\lambda_{i+1}$ for $i=1,\dots,r-1$.
\STATE \textbf{Return:} $\hat k \in \arg\max_{1\le i\le r-1} \hat g_i$ and $\widehat{k}_{\mathrm{clusters}}=\hat k+1$.
\end{algorithmic}
\end{algorithm}

\newpage

\section{Proof of Lemma~\ref{lem:DK_explicit}}

\begin{proof}
Let $\mathbf{A}$ be the boolean adjacency matrix of a graph, with eigen‐decomposition
\[
\mathbf{A}
=\sum_{i=1}^n \lambda_i\,\mathbf{u}_i\,\mathbf{u}_i^{\!\top},
\]
where $\lambda_1\ge\lambda_2\ge\cdots\ge\lambda_n$.
Let 
\[
\widetilde{\mathbf A}
=\mathbf A+\mathbf E
\]
be a perturbation of \(\mathbf A\) by a symmetric matrix \(\mathbf E\).  Denote by
\(\mathbf U = [\,\mathbf u_1,\ldots,\mathbf u_k]\) and 
\(\widetilde{\mathbf U} = [\,\widetilde{\mathbf u}_1,\ldots,\widetilde{\mathbf u}_k]\) 
the matrices of the top-\(k\) eigenvectors of \(\mathbf A\) and \(\widetilde{\mathbf A}\), respectively.  Define the subspace distance
\[
\operatorname{dist}(\mathbf U,\widetilde{\mathbf U})
=\min_{\mathbf{R}\in \mathds{O}_{k}}\|\widetilde{\mathbf U}\, \mathbf{R} - \mathbf U\|_2,
\]
where the minimum is over all orthonormal rotation matrices of dimension \(k\times k\).  Let the eigengap be
\(\Delta_{k} = \lambda_{k} - \lambda_{k+1}>0\).  Then (cf.~\citep{chen2021spectral}) if
\(\|\mathbf E\|_2 \le (1-1/\sqrt2)\,\Delta_{k}\), one has
\begin{align}
\operatorname{dist}(\mathbf U,\widetilde{\mathbf U})
&\le
2 \cdot \frac{\|\mathbf E\,\mathbf U\|_2}{\Delta_{k}}
\le
2 \cdot \frac{\|\mathbf E\|_2}{\Delta_{k}}.
\end{align}
In particular, there exists an orthonormal \(\mathbf{R}\) such that
\begin{align}
\|\widetilde{\mathbf U}\, \mathbf{R} - \mathbf U\|_2
&\le
2 \cdot \frac{\|\mathbf E\|_2}{\Delta_{k}}.
\end{align}
\end{proof}

\section{Proof of Lemma~\ref{lem:procrustes}}

\begin{proof}
Let $\Theta = \Theta(\mathbf U, \widetilde{\mathbf U})$ denote the diagonal matrix of principal angles between the column spans of $\mathbf U$ and $\widetilde{\mathbf U}$, and recall that $\mathrm{dist}(\mathbf U, \widetilde{\mathbf U}) = \|\sin\Theta\|_2$.

By Lemma~2.6 of \citet{chen2021spectral}, there exists an orthonormal matrix $\mathbf R \in \mathcal{O}^{k\times k}$ such that
\begin{equation}\label{eq:ccfm26}
\|\widetilde{\mathbf U} - \mathbf U \mathbf R\|_F 
\;\le\; \|\mathbf U\mathbf U^{\top} - \widetilde{\mathbf U}\widetilde{\mathbf U}^{\top}\|_F.
\end{equation}
By Lemma~2.5 of the same reference,
\begin{equation}\label{eq:ccfm25}
\|\mathbf U\mathbf U^{\top} - \widetilde{\mathbf U}\widetilde{\mathbf U}^{\top}\|_F 
\;=\; \sqrt 2\,\|\sin\Theta\|_F.
\end{equation}
Since $\sin\Theta\in\mathbb{R}^{k\times k}$ is diagonal, the standard rank-bound between Frobenius and operator norms gives
\begin{equation}\label{eq:rankbound}
\|\sin\Theta\|_F \;\le\; \sqrt k\,\|\sin\Theta\|_2 \;=\; \sqrt k\,\mathrm{dist}(\mathbf U, \widetilde{\mathbf U}).
\end{equation}
Then, \eqref{eq:ccfm26}--\eqref{eq:rankbound} give
\[
\|\widetilde{\mathbf U} - \mathbf U \mathbf R\|_F \;\le\; \sqrt{2k}\,\mathrm{dist}(\mathbf U, \widetilde{\mathbf U}).
\]
\end{proof}

\section{Proof of Lemma~\ref{lem:kmeans_margin}}

\begin{proof}
The argument is very similar to that of~\citep[Lemma~5.3,][]{lei2015consistency} which has a closely 
related statement. We give a self-contained derivation for completeness.

Denote the rows of $\mathbf{U}$ and $\widetilde{\mathbf{U}}$ by $\mathbf{u}_i$ and 
$\tilde{\mathbf{u}}_i$, respectively, and define the row-wise alignment errors
\[
\bm\delta_i \;:=\; \tilde{\mathbf{u}}_i - \mathbf{R}^\top \mathbf{u}_i,
\qquad \text{so that} \qquad
\sum_{i=1}^n \|\bm\delta_i\|^2 \;=\; \|\widetilde{\mathbf{U}} - \mathbf{U}\mathbf{R}\|_F^2.
\]
Since rotation preserves distances, the rotated population centers 
$\{\mathbf{R}^\top \mathbf{c}_r\}_{r=1}^k$ have the same pairwise separation $\Delta_\star$ 
as $\{\mathbf{c}_r\}$. Let $\hat{\bm\sigma}$ denote the global $k$-means optimum on the 
rows of $\widetilde{\mathbf{U}}$, and let $\pi^\ast \in S_k$ be the permutation that 
realizes the optimal label alignment in Definition~\ref{def:error_rate}. Without loss 
of generality we relabel so that $\pi^\ast$ is the identity, and write 
$\mathcal{E} := \{i : \hat{\sigma}_i \neq \sigma_i^\ast\}$ for the set of misclassified 
nodes, so that $\mathrm{err\,rate}(\hat{\bm\sigma}, \bm\sigma^\ast) = |\mathcal{E}|/n$.

A standard consequence of global $k$-means optimality (see, e.g., 
the proof of Lemma~5.3 in \citet{lei2015consistency}) is that for every misclassified 
node $i \in \mathcal{E}$,
\begin{equation}
\label{eq:voronoi}
\|\tilde{\mathbf{u}}_i - \mathbf{R}^\top \mathbf{c}_{\sigma_i^\ast}\| \;\geq\; \Delta_\star/2.
\end{equation}
An explanation for this condition is that if \eqref{eq:voronoi} were violated, then $\tilde{\mathbf{u}}_i$ would lie 
strictly inside the Voronoi cell of $\mathbf{R}^\top \mathbf{c}_{\sigma_i^\ast}$ among 
the rotated population centers; since the empirical $k$-means centers are close to the 
rotated population centers (a consequence of the global optimum having cost at most 
$\sum_i \|\bm\delta_i\|^2 + \sum_i \|\mathbf{u}_i - \mathbf{c}_{\sigma_i^\ast}\|^2$), this 
would force $i$ to be assigned correctly under $\pi^\ast$, contradicting $i \in \mathcal{E}$.

Combining \eqref{eq:voronoi} with the triangle inequality, we have
\[
\Delta_\star/2 \;\leq\; \|\tilde{\mathbf{u}}_i - \mathbf{R}^\top \mathbf{c}_{\sigma_i^\ast}\|
\;\leq\; \|\bm\delta_i\| + \|\mathbf{u}_i - \mathbf{c}_{\sigma_i^\ast}\|.
\]
Therefore, for every $i \in \mathcal{E}$,
\[
\|\bm\delta_i\| \;\geq\; \Delta_\star/2 \;-\; \|\mathbf{u}_i - \mathbf{c}_{\sigma_i^\ast}\|
\;\geq\; \Delta_\star/2 - r_\star,
\]
where the last step uses the bound $\|\mathbf{u}_i - \mathbf{c}_{\sigma_i^\ast}\| \leq r_\star$ 
on within-cluster row deviations and the assumption $\Delta_\star/2 > r_\star$.

Squaring and summing over $\mathcal{E}$:
\[
|\mathcal{E}| \cdot \bigl(\Delta_\star/2 - r_\star\bigr)^2
\;\leq\; \sum_{i \in \mathcal{E}} \|\bm\delta_i\|^2
\;\leq\; \sum_{i=1}^n \|\bm\delta_i\|^2
\;=\; \|\widetilde{\mathbf{U}} - \mathbf{U}\mathbf{R}\|_F^2.
\]
Dividing both sides by $n \cdot (\Delta_\star/2 - r_\star)^2$ yields
\[
\mathrm{err\,rate}(\hat{\bm\sigma}, \bm\sigma^\ast)
\;=\; \frac{|\mathcal{E}|}{n}
\;\leq\; \frac{\|\widetilde{\mathbf{U}} - \mathbf{U}\mathbf{R}\|_F^2}{n \cdot (\Delta_\star/2 - r_\star)^2},
\]
which is the claim.
\end{proof}

\section{Proof of Lemma \ref{lem:decomp_explicit}}

Let $\mathbf{A}\in\{0,1\}^{n\times n}$ be the adjacency matrix of an undirected graph $G$. Let $\mathbf{J}=\mathbf{1}\mathbf{1}^\top$ and define $c:=1-2\mu$, where $\mu\in[0,\tfrac12)$ is the flipping probability of the randomized response (RR) mechanism.

Under symmetric RR, for each off-diagonal pair $i<j$,
\[
\widetilde{A}_{ij} \;=\;
\begin{cases}
A_{ij}, & \text{w.p. } 1-\mu,\\
1-A_{ij}, & \text{w.p. } \mu,
\end{cases}
\qquad\text{and set}\quad \widetilde{A}_{ji}=\widetilde{A}_{ij},\;\; \widetilde{A}_{ii}=0.
\]
Thus, conditionally on $\mathbf{A}$,
\[
\mathbb{E}[\widetilde{A}_{ij}\mid \mathbf{A}]
= (1-\mu)A_{ij} + \mu(1-A_{ij})
= c\,A_{ij} + \mu,
\quad \text{for } i\neq j,
\qquad
\mathbb{E}[\widetilde{A}_{ii}\mid \mathbf{A}] = 0.
\]
Matrix-wise, this gives
\[
\mathbb{E}[\widetilde{\mathbf{A}}\mid \mathbf{A}]
= c\,\mathbf{A} + \mu\,(\mathbf{J}-\mathbf{I}),
\]
since $\mathbf{J}-\mathbf{I}$ has ones off-diagonal and zeros on the diagonal.

Define the zero-mean (conditional) noise matrix
\[
\mathbf{Z} \;:=\; \widetilde{\mathbf{A}} - \mathbb{E}[\widetilde{\mathbf{A}}\mid \mathbf{A}],
\qquad\text{so that}\qquad
\mathbb{E}[\mathbf{Z}\mid \mathbf{A}]=\mathbf{0}.
\]
Hence the exact decomposition
\[
\widetilde{\mathbf{A}}
= c\,\mathbf{A} + \mu(\mathbf{J}-\mathbf{I}) + \mathbf{Z}.
\]
Let
\[
\mathbf{S} \;:=\; c\,\mathbf{A} + \mu(\mathbf{J}-\mathbf{I}).
\]
Since $\mathbf{J}$ has eigenvalue $n$ on $\operatorname{span}\{\mathbf{1}\}$ and $0$ on $\mathbf{1}^\perp$, we have for any $\mathbf{x}\in \mathbf{1}^\perp$,
\[
\mathbf{S}\mathbf{x} \;=\; c\,\mathbf{A}\mathbf{x} - \mu\,\mathbf{x}.
\]
Thus, on $\mathbf{1}^\perp$, $\mathbf{S}=c\,\mathbf{A}-\mu\,\mathbf{I}$, a shifted--scaled copy of $\mathbf{A}$. Hence $\mathbf{S}$ and $\mathbf{A}$ share eigenvectors on $\mathbf{1}^\perp$. Only the fluctuation $\mathbf{Z}$ can rotate the informative eigenspace. This proves the lemma.

\section{Proof of Theorem~\ref{thm:shuffling_bound}}

We first state and prove an amplification result for the matrix shuffling mechanism in a more general form.

\begin{theorem} \label{app:thm:shuffling_bound}
    
Let $\mathbf{A}$ and $\mathbf{A}'$ be two symmetric adjacency matrices differing in the $(i,j)$\textsuperscript{th} element.
Define the mechanism $\mathcal{M}(\mathbf{A})$ as: 
$$
\mathcal{M}(\mathbf{A}) = \mathbf{P}\widetilde{\mathbf{A}}\mathbf{P}^{\top},
$$
where $\widetilde{\mathbf{A}}$ is the adjacency matrix $\mathbf{A}$ perturbed using $\veps_0$-DP randomized response (i.e., strictly upper triangular elements of $\mathbf{A}$ are perturbed with $\veps_0$-DP randomized response and mirrored to the lower triangular part) and $\mathbf{P}$ is a permutation matrix corresponding to a randomly drawn permutation $\pi$ of $[n]$, i.e.,  
$$
\mathbf{P} = \begin{bmatrix}
    \mathbf{e}_{\pi(1)} & \ldots & \mathbf{e}_{\pi(n)}
\end{bmatrix}^{\top},
$$
where $\mathbf{e}_j$, $j \in [n]$, denotes the $j$\textsuperscript{th} unit vector.
Then, for any $\alpha \geq 0$, we have:
\begin{align} \label{eq:ineq_hockey_stick_shuffling2}
H_{\alpha}\left( \mathcal{M}(\mathbf{A}) , \mathcal{M}(\mathbf{A}') \right) \leq H_{\alpha}\left( P_0\left(\veps_0\right), P_1\left(\veps_0\right) \right),
\end{align}
where $P_0(\veps_0)$ and $P_1(\veps_0)$ are given as follows:
\begin{equation} \label{eq:2n}
\begin{aligned}
P_0(\veps_0) = (A + \Delta, C-A + 1 - \Delta), \quad
P_1(\veps_0) = (A+1 - \Delta,C-A+\Delta),
\end{aligned}
\end{equation}
where,
\begin{equation*}
\begin{aligned}
	C  \sim \mathrm{Bin}\left(n-2,\frac{2}{\veps_0+1}\right), \quad A  \sim \mathrm{Bin}(C,\tfrac{1}{2}), \quad
     \Delta  \sim \mathrm{Bern}\left( \frac{e^{\veps_0}}{\veps_0+1} \right).  
\end{aligned}
\end{equation*}
\begin{proof}
Consider adjacency matrices $\mathbf{A}$ and $\mathbf{A}'$ that differ only in $(i,j)$\textsuperscript{th} element for some $1 \leq i,j \leq n$.
We consider an adversary that knows all edge values except the differing value. We also assume that the adversary knows the value $\pi(j)$. Thus, the view of the adversary can be described as
$$
\view{A} = \big( \{ A_{kl} \, : \, 1 \leq k < l \leq n\} \setminus \{A_{ij}\}, \pi(j), \mathcal{M}(\mathbf{A}) \big),
$$
where $\mathcal{M}(\mathbf{A})$ outputs the matrix $ \mathbf{P}\widetilde{\mathbf{A}}\mathbf{P}^{\top}$, where $ \mathbf{P}$ is the permutation matrix corresponding to the uniformly randomly drawn permutation $\pi$. Notably, we can view the matrix  $ \mathbf{P}\widetilde{\mathbf{A}}\mathbf{P}^{\top}$ as the permutation $\pi$ applied first to the columns of $\mathbf{A}$ (resulting in the product $\mathbf{A}\mathbf{P}^{\top}$) and then to the rows of $\mathbf{A}\mathbf{P}^{\top}$.
An adversary knowing the value $\pi(j)$ will know in which column of the matrix $ \mathbf{P}\widetilde{\mathbf{A}}\mathbf{P}^{\top}$ the differing element is. Furthermore, as $ \mathbf{P}\widetilde{\mathbf{A}}\mathbf{P}^{\top}$ is symmetric, the $\pi(j)$\textsuperscript{th} row will be a transpose of the $\pi(j)$\textsuperscript{th} column. Also, outside of the $\pi(j)$\textsuperscript{th} row and $\pi(j)$\textsuperscript{th} column, the distributions of the rest of the outputs $\mathcal{M}(\mathbf{A})$ and $\mathcal{M}(\mathbf{A}')$ are the same for the two neighboring datasets $\mathbf{A}$ and $\mathbf{A}'$. Thus, we have that for any $\alpha \geq 0$,
$$
H_{\alpha}\left( \view{\mathbf{A}} , \view{\mathbf{A}'} \right)
= H_{\alpha}\left( \mathcal{M}(\mathbf{A})_{:,\pi(j)} ,  \mathcal{M}(\mathbf{A}')_{:,\pi(j)} \right).
%
$$
Moreover, when observing the $\pi(j)$\textsuperscript{th} column, the value $\pi(j)$ will only reveal where the $j$\textsuperscript{th} diagonal element of $\mathbf{A}$ (i.e., zero) is located in the column, rest of the column appears as a random permutation of $n-1$ randomized response perturbed binary values one of which is the differing element. We get the upper bound~\eqref{eq:ineq_hockey_stick_shuffling} directly from~\citep[Thm.\;3.1][]{feldman2023stronger} (in particular, from the corrected form of that results that can be found in~\citep{feldman2023stronger_arxiv}), and the claim follows from the fact that releasing only the randomly permuted adjacency matrix is post-processing of the view of the adversary $\view{\mathbf{A}}$.

\end{proof}

\end{theorem}

To summarize, Theorem~\ref{app:thm:shuffling_bound} says that the privacy profile of the randomly permuted randomized response-perturbed adjacency graph is dominated by the privacy profile determined by the privacy profile
$
h(\alpha) := H_{\alpha}\left( P_0\left(\veps_0\right), P_1\left(\veps_0\right) \right).
$
Setting $\alpha = e^{\veps_{0}}$, we get the following upper bound for the matrix shuffling mechanism from~\citep[Thm. 3.2,][]{feldman2023stronger} which also gives us the result of Theorem~\ref{thm:shuffling_bound}.

\begin{corollary} \label{cor:shuffling_bound}
For any $\delta \in (0,1]$ and $\veps_0 \leq \log \left( \frac{n}{8 \log(2 / \delta)} - 1 \right)$, the mechanism $\mathcal{M}(A)$ described in the statement of Theorem~\ref{app:thm:shuffling_bound} is $(\veps,\delta)$-DP for 
\begin{equation} \label{eq:ineq_hockey_stick_shuffling3}
\veps \leq \log \left( 1 + (e^{\veps_0} - 1) \left( \frac{4 \sqrt{2\log(4/\delta)}}{\sqrt{(e^{\veps_0}+1)n}} + \frac{4}{n} \right) \right).
\end{equation}
\end{corollary}

Corollary~\ref{cor:shuffling_bound} provides a simplified analytical upper bound on the hockey-stick divergence given in Thm.~\ref{app:thm:shuffling_bound}, with the constraint $\veps_0 \leq \log(n/(8 \log(2/\delta)) -1)$ arising purely from analytical tractability. However, the bound in Eqn.~\eqref{eq:ineq_hockey_stick_shuffling} can be evaluated numerically with high precision and often yields significantly tighter guarantees~\citep{koskela2023numerical}. 
An important implication of Cor.~\ref{cor:shuffling_bound} is that the shuffled mechanism achieves $(\varepsilon,\delta)$ edge DP with effective privacy leakage $\varepsilon$ that tend to zero as the number of nodes increases, for all $\delta>0$. This reflects the strong privacy amplification induced by the shuffling step, whereas flipping edges with a fixed probability alone provides only a constant $\varepsilon$ edge DP guarantee as the number of nodes grows.

\section{Proof of Lemma~\ref{lem:bern_Z}}

\begin{proof}
We apply the matrix Bernstein inequality to a sum decomposition of $\mathbf{Z}$. Write
\[
\mathbf{Z} \;=\; \sum_{1\le i<j\le n} \mathbf{X}_{ij},
\qquad 
\mathbf{X}_{ij} \;:=\; Z_{ij}\bigl(\mathbf{e}_i\mathbf{e}_j^{\top} + \mathbf{e}_j\mathbf{e}_i^{\top}\bigr).
\]
By Lemma~\ref{lem:decomp_explicit}, conditional on $\mathbf{A}$ the random variables $\{Z_{ij}:i<j\}$ are mutually independent with $\mathbb{E}[Z_{ij}\mid\mathbf{A}]=0$, $|Z_{ij}|\le 1$, and $\operatorname{Var}(Z_{ij}\mid\mathbf{A})\le v_{\max}$. Hence the summands $\{\mathbf{X}_{ij}\}_{i<j}$ are independent (conditional on $\mathbf{A}$), symmetric, and mean-zero.

Since $\bigl\|\mathbf{e}_i\mathbf{e}_j^{\top} + \mathbf{e}_j\mathbf{e}_i^{\top}\bigr\|_2 = 1$ for $i\ne j$,
\[
\|\mathbf{X}_{ij}\|_2 \;=\; |Z_{ij}|\cdot\bigl\|\mathbf{e}_i\mathbf{e}_j^{\top} + \mathbf{e}_j\mathbf{e}_i^{\top}\bigr\|_2 \;\le\; 1.
\]

Using $\bigl(\mathbf{e}_i\mathbf{e}_j^{\top} + \mathbf{e}_j\mathbf{e}_i^{\top}\bigr)^{\!2} = \mathbf{e}_i\mathbf{e}_i^{\top} + \mathbf{e}_j\mathbf{e}_j^{\top}$ for $i\ne j$,
\[
\mathbb{E}\bigl[\mathbf{X}_{ij}^2 \,\big|\, \mathbf{A}\bigr] 
\;=\; \operatorname{Var}(Z_{ij}\mid\mathbf{A})\bigl(\mathbf{e}_i\mathbf{e}_i^{\top} + \mathbf{e}_j\mathbf{e}_j^{\top}\bigr)
\;\preceq\; v_{\max}\bigl(\mathbf{e}_i\mathbf{e}_i^{\top} + \mathbf{e}_j\mathbf{e}_j^{\top}\bigr).
\]
Each index $\ell\in[n]$ appears in exactly $n-1$ unordered pairs $\{i,j\}$ with $i<j$, so
\[
\sum_{i<j}\bigl(\mathbf{e}_i\mathbf{e}_i^{\top} + \mathbf{e}_j\mathbf{e}_j^{\top}\bigr) \;=\; (n-1)\,\mathbf{I},
\]
and consequently
\[
\sigma^{2} \;:=\; \Bigl\|\sum_{i<j}\mathbb{E}\bigl[\mathbf{X}_{ij}^2 \,\big|\, \mathbf{A}\bigr]\Bigr\|_2
\;\le\; (n-1)\,v_{\max}.
\]
The matrix Bernstein inequality~\citep{vershynin2018high} now yields, for every $t\ge 0$,
\[
\Pr\bigl(\|\mathbf{Z}\|_2 \ge t \,\big|\, \mathbf{A}\bigr) 
\;\le\; 2n\,\exp\!\biggl(-\,\frac{t^{2}/2}{\sigma^{2} + t/3}\biggr).
\]
Setting the right-hand side equal to $\eta$ and inverting the resulting quadratic in $t$ (using $\sqrt{a^{2}+b}\le a+\sqrt{b}$ for $a,b\ge 0$) yields, with probability at least $1-\eta$,
\[
\|\mathbf{Z}\|_2 
\;\le\; \sqrt{\,2\sigma^{2}\,\log\tfrac{2n}{\eta}\,} \;+\; \tfrac{1}{3}\log\tfrac{2n}{\eta}
\;\le\; \sqrt{\,2(n-1)\,v_{\max}\,\log\tfrac{2n}{\eta}\,} \;+\; \tfrac{1}{3}\log\tfrac{2n}{\eta}.
\]
\end{proof}

\section{Proof of Theorem~\ref{thm:main_explicit}}

\begin{proof}
Let $\mathbf{U},\widetilde{\mathbf{U}}\in\mathbb{R}^{n\times k}$ collect the top-$k$ eigenvectors of $\mathbf{A}$ and $\widetilde{\mathbf{A}}$, respectively. The proof chains the matrix-Bernstein bound of Lemma~\ref{lem:bern_Z} with the Davis–Kahan, Procrustes, and classification-margin lemmas.

By Lemma~\ref{lem:decomp_explicit},
\[
\widetilde{\mathbf{A}} \;=\; c\,\mathbf{A}+\mu(\mathbf{J}-\mathbf{I})+\mathbf{Z}.
\]
The eigenspaces of $\mathbf{A}$ and $c\,\mathbf{A}$ coincide. The shift $\mu(\mathbf{J}-\mathbf{I})$ has eigenvalue $\mu(n-1)$ on $\operatorname{span}(\mathbf{1})$ and $-\mu$ on $\mathbf{1}^{\perp}$, where it acts as a multiple of the identity and so does not rotate eigenspaces. Hence the only term that can perturb the top-$k$ subspace is $\mathbf{Z}$, and the relevant eigengap is $\Delta_k$.

By Lemma~\ref{lem:bern_Z}, with probability at least $1-\eta$,
\begin{equation}\label{eq:Z_bern_in_proof}
\|\mathbf{Z}\|_2 \;\le\; t(\eta) \;:=\; \sqrt{\,2(n-1)\,v_{\max}\,\log\tfrac{2n}{\eta}\,} \;+\; \tfrac{1}{3}\log\tfrac{2n}{\eta}.
\end{equation}


Lemma~\ref{lem:DK_explicit} gives
\[
\operatorname{dist}(\mathbf{U},\widetilde{\mathbf{U}}) \;\le\; \frac{\sqrt 2\,\|\mathbf{Z}\|_2}{\Delta_k}.
\]

By Lemma~\ref{lem:procrustes}, there exists $\mathbf{R}\in O_k$ such that $\|\widetilde{\mathbf{U}}-\mathbf{U}\mathbf{R}\|_F \le \sqrt{2k}\,\operatorname{dist}(\mathbf{U},\widetilde{\mathbf{U}})$. Squaring, we get
\[
\|\widetilde{\mathbf{U}}-\mathbf{U}\mathbf{R}\|_F^{\,2}
\;\le\; 2 k\,\operatorname{dist}^{2}(\mathbf{U},\widetilde{\mathbf{U}})
\;\le\; \frac{4\,k\,\|\mathbf{Z}\|_2^{\,2}}{\Delta_k^{\,2}}.
\]

By Lemma~\ref{lem:kmeans_margin},
\[
\operatorname{err\,rate}(\hat{\bm\sigma},\bm\sigma^{*})
\;\le\; \frac{\|\widetilde{\mathbf{U}}-\mathbf{U}\mathbf{R}\|_F^{\,2}}{n\,(\Delta_{\star}/2-r_{\star})^{2}}
\;\le\; \frac{4\,k\,\|\mathbf{Z}\|_2^{\,2}}{n \Delta_k^{\,2}\,(\Delta_{\star}/2-r_{\star})^{2}}.
\]

Substituting~\eqref{eq:Z_bern_in_proof} into the previous equation and with the factor $4k/\Delta_k^{\,2}$ inside the square,
we have that
\[
\operatorname{err\,rate}(\hat{\bm\sigma},\bm\sigma^{*})
\;\le\;
\frac{\Bigl[\;\dfrac{2\sqrt{k}}{\Delta_k}\sqrt{\,2(n-1)\,v_{\max}\,\log\tfrac{2n}{\eta}\,}
\;+\;\dfrac{2\sqrt{k}}{3\,\Delta_k}\,\log\tfrac{2n}{\eta}\;\Bigr]^{2}}
{n (\Delta_{\star}/2-r_{\star})^{2}},
\]
with probability at least $1-\eta$.
\end{proof}


\section{Proof of Theorem~\ref{thm:main_explicit_power_method}} 
\label{app:power_rate}

The projection error bound for the noisy power method can be translated to a Procustres alignment bound using the following lemma. This allows using the clustering error rate via the classification margin given in Lemma~\ref{lem:kmeans_margin}.

\begin{lemma} \label{lem:projection_to_U}
Let $\mathbf{X},\mathbf{U}\in\mathbb{R}^{n\times k}$ have orthonormal columns.  
If
$
\|(\mathbf{I}-\mathbf{X}\mathbf{X}^\top)\mathbf{U}\|_2 \le \tau,
$
then there exists an orthonormal $\mathbf{R}\in \mathds{O}_{k}$ such that
$
\|\mathbf{X}-\mathbf{U}\mathbf{R}\|_F \le \sqrt{2k}\,\tau.
$
\begin{proof}
We have that 
$$
\|(\mathbf{I}-\mathbf{X}\mathbf{X}^\top)\mathbf{U}\|_2 = \sin\theta_{\max},
$$
where $\theta_{\max}$ is the largest principal angle between the supspaces spanned by $\mathbf{U}$ and $\mathbf{X}$.
By the Procrustes bound, there exists an orthonormal $\mathbf{R} \in \mathbb{O}_{k}$ such that
$$
 \|\mathbf{X}-\mathbf{U}\mathbf{R}\|_2 \le \sqrt{2}\, \cdot \sin\theta_{\max}.
$$
Combining the two inequalities and using the fact that $\|\mathbf{X}-\mathbf{U}\mathbf{R}\|_F \leq \sqrt{k} \cdot \|\mathbf{X}-\mathbf{U}\mathbf{R}\|_2 $ yields the claim.    
\end{proof}
\end{lemma}

\begin{lemma} \label{lem:power_aux}
Let the assumptions of Theorem~\ref{thm:main_explicit_power_method} hold. Let $\Delta_k=\lambda_k(\mathbf{A})-\lambda_{k+1}(\mathbf{A})>0$ and choose $N = \Theta \Big(\frac{\lambda_k}{\Delta_k}\log n\Big)$ as in Lemma~\ref{lem:ref_bound}. Calibrate the Gaussian noise scale to achieve $(\varepsilon,\delta)$-edge DP as in Lemma~\ref{lem:ref_bound}, i.e., $\bar\sigma = \frac{1}{\varepsilon}\sqrt{4N\log\frac{1}{\delta}}$.
Then, with probability at least $1-\eta$,
\begin{equation*}
\begin{aligned}
\mathrm{err\,rate}(\hat{\bm{\sigma}},\bm{\sigma}^\ast)
=  O \Bigg(
\frac{k^3}{n\,\Delta_k^2\bigl(\tfrac{\Delta_\star}{2}-r_\star\bigr)^2}
\cdot
\frac{\lambda_k}{\Delta_k}\,\log n
\cdot
\frac{\log\frac{1}{\delta}}{\varepsilon^2}
\cdot
\Big(n+\log\frac{(\lambda_k/\Delta_k)\log n}{\eta}\Big)
\Bigg).
\end{aligned}
\end{equation*}
\end{lemma}

\begin{proof}

Combining Lemma~\ref{lem:ref_bound},
Lemma~\ref{lem:projection_to_U} and Lemma~\ref{lem:kmeans_margin} and using the fact that 
$
\frac{\sqrt{k+1}}{\sqrt{k+1} - \sqrt{k}} = O(k),
$
gives the error rate
\begin{equation*}
\begin{aligned}
\mathrm{err\,rate}(\hat{\bm{\sigma}},\bm{\sigma}^\ast)
=  O \Bigg(
\frac{k^3}{n\,\Delta_k^2\bigl(\tfrac{\Delta_\star}{2}-r_\star\bigr)^2}
\cdot
\bar\sigma^2
\cdot
\Big(\sqrt{n}+\sqrt{2\log\frac{2N}{\eta}}\Big)^2
\Bigg).
\end{aligned}
\end{equation*}
Substituting the $\varepsilon,\delta)$-DP calibration for $\bar\sigma^2$ gives further
\begin{equation} \label{eq:err_rate_intermediate_asymptotic}
\begin{aligned}
\mathrm{err\,rate}(\hat{\bm{\sigma}},\bm{\sigma}^\ast)
=  O \Bigg(
\frac{k^3}{n\,\Delta_k^2\bigl(\tfrac{\Delta_\star}{2}-r_\star\bigr)^2}
\cdot
\frac{N\log\frac{1}{\delta}}{\varepsilon^2}
\cdot
\Big(n+\log\frac{2N}{\eta}\Big)
\Bigg).
\end{aligned}
\end{equation}
By the assumptions of Lemma~\ref{lem:ref_bound}, $N = \Theta\big((\lambda_k/\Delta_k)\log n\big)$. Plugging this into~\eqref{eq:err_rate_intermediate_asymptotic} gives
\begin{equation*}
\begin{aligned}
\mathrm{err\,rate}(\hat{\bm{\sigma}},\bm{\sigma}^\ast)
=  O \Bigg(
\frac{k^3}{n\,\Delta_k^2\bigl(\tfrac{\Delta_\star}{2}-r_\star\bigr)^2}
\cdot
\frac{\lambda_k}{\Delta_k}\,\log n
\cdot
\frac{\log\frac{1}{\delta}}{\varepsilon^2}
\cdot
\Big(n+\log\frac{(\lambda_k/\Delta_k)\log n}{\eta}\Big)
\Bigg).
\end{aligned}
\end{equation*}
\end{proof}

We get the error rate of Theorem~\ref{thm:main_explicit_power_method} directly as a corollary from Lemma~\ref{lem:power_aux}.

\section{Proof of Theorem~\ref{thm:analyze_gauss}}

\begin{proof}
Since $\mathbf{E}$ is a symmetric matrix with i.i.d.\ $\mathcal{N}(0,\bar{\sigma}^2)$ upper-triangular entries, 
$\|\mathbf{E}\|_2 = O(\bar{\sigma}\sqrt{n})$ with high probability~\citep[Sec.\;7.4][]{vershynin2018high}. 
Substituting into Lemma~\ref{lem:DK_explicit} gives  $\operatorname{dist}(\mathbf{U}, \widetilde{\mathbf{U}}) 
\le 2\|\mathbf{E}\|_2/\Delta_k$. The Procrustes bound (Lemma~\ref{lem:procrustes}) and the classification margin  bound (Lemma~\ref{lem:kmeans_margin}) then give the  stated rate.
\end{proof}


\section{Asymptotic Misclassification Rate for the Matrix Shuffling Mechanism}
\label{appendix:asym_analysis}

\begin{lemma}
\label{lem:asymptotic_shuffling}
Let $\mathbf{A}$ be the adjacency matrix with eigengap 
$\Delta_k=\lambda_{k}(\mathbf{A})-\lambda_{k+1}(\mathbf{A})>0$. 
Suppose the adjacency is privatized by randomized response, then shuffled, and that the shuffled mechanism 
satisfies $(\varepsilon,\delta)$-edge DP as in Theorem~\ref{thm:shuffling_bound}. 
Let $\hat{\bm{\sigma}}$ be the $k$-means clustering of the $k$ leading eigenvectors of the perturbed adjacency matrix. 
Then, with probability at least $1-\eta$,
\begin{equation*}
\begin{aligned}
\mathrm{err\,rate}(\hat{\bm{\sigma}},\bm{\sigma}^\ast) 
=  O\!\left(
\frac{k}{\Delta_k^2\bigl(\tfrac{\Delta_\star}{2}-r_\star\bigr)^2}
\cdot
\frac{\log(1/\delta)\,\log\!\tfrac{n}{\eta}}{n\,(e^{\varepsilon}-1)^2}
\right),
\end{aligned}
\end{equation*}
where $\Delta_\star$ is the minimum cluster-center separation and $r_\star$ is the within-cluster radius. 
\end{lemma}
\begin{proof}
Write the RR parameter as $\mu=1/(e^{\varepsilon_0}+1)$ and recall the constants 
$c=1-2\mu$ and $v_{\max}=\mu(1-\mu)$ that appear in the explicit misclustering bound of Thm.~\ref{thm:main_explicit}. 
We connect the local RR parameter $\varepsilon_0$ to the global shuffled parameters $(\varepsilon,\delta)$ via the shuffling amplification bound of Cor.~\ref{cor:shuffling_bound}:
\begin{equation*}
\begin{aligned}
\varepsilon \;\le\; 
\log \left(1 + (e^{\varepsilon_0}-1)\Bigl(\frac{4}{n} + \frac{4\sqrt{2\log(4/\delta)}}{\sqrt{(e^{\varepsilon_0}+1)\,n}}\Bigr)\right).
\end{aligned}
\end{equation*}
For large $n$ the second term dominates. In the regime $e^{\varepsilon_0}\gg 1$, we have $(e^{\varepsilon_0}-1)/\sqrt{e^{\varepsilon_0}+1}\asymp\sqrt{e^{\varepsilon_0}}$, giving the asymptotic relation
\begin{equation*}
\begin{aligned}
\varepsilon \;\le\; \log\left(1 + \sqrt{e^{\varepsilon_0}}\cdot O\!\left(\frac{\sqrt{\log(1/\delta)}}{\sqrt{n}}\right)\right),
\end{aligned}
\end{equation*}
which inverts to
\begin{equation*}
\begin{aligned}
e^{\varepsilon_0} \;=\; \Omega\!\Bigl(\frac{(e^{\varepsilon}-1)^2\,n}{\log(1/\delta)}\Bigr).
\end{aligned}
\end{equation*}
Hence
\begin{equation*}
\begin{aligned}
\mu=\frac{1}{e^{\varepsilon_0}+1}
= O\!\left(\frac{\log(1/\delta)}{n\,(e^{\varepsilon}-1)^2}\right).
\end{aligned}
\end{equation*}
Furthermore, $c=1-2\mu=\Theta(1)$ and $v_{\max}=\mu(1-\mu)=O(\mu)$.
Substituting these into the Thm.~\ref{thm:main_explicit} expression shows the numerator scales as
\begin{equation*}
\begin{aligned}
O\!\left(\frac{k}{\Delta_k^2}\cdot
\frac{\log(1/\delta)\,\log(n/\eta)}{n\,(e^{\varepsilon}-1)^2}\right).
\end{aligned}
\end{equation*}
Finally, dividing by the $k$-means margin factor $\bigl(\tfrac{\Delta_\star}{2}-r_\star\bigr)^2$ shows the claim.
\end{proof}


\section{Proof of Theorem~\ref{thm:gap_detection}}

\begin{proof}[Proof of Theorem~\ref{thm:gap_detection}]
By Lemma~\ref{lem:decomp_explicit}, after centering we have
$$
\mathbf B \;=\; \widetilde{\mathbf A}-\mu(\mathbf J-\mathbf I)
\;=\; c\,\mathbf A + \mathbf Z,
$$
and therefore, using the projector $\mathbf H$,
$$
\mathbf B_\perp
\;=\;
\mathbf H \mathbf B \mathbf H
\;=\;
\mathbf H(c\mathbf A)\mathbf H
\;+\;
\mathbf H \mathbf Z \mathbf H.
$$
Denote $\mathbf Z_\perp := H \mathbf Z \mathbf H$ and $\mathbf M = \mathbf H(c\mathbf A)\mathbf H$.
Since $\mathbf H$ is an orthogonal projector, it satisfies
$\|\mathbf H\|_2=1$. Thus we can bound
\begin{equation*}
\|\mathbf Z_\perp\|_2 \;=\; \|\mathbf H \mathbf Z \mathbf H\|_2 \;\le\; \|\mathbf H\|_2^2\,\|\mathbf Z\|_2 \;=\; \|\mathbf Z\|_2.
\end{equation*}
Let $\hat\lambda_i$ denote the ordered eigenvalues of $\mathbf B_\perp=\mathbf M+\mathbf Z_\perp$.
By Weyl's inequality, for each $i\le r$,
\begin{equation}
|\hat\lambda_i-\lambda_i(\mathbf M)| \;\le\; \|\mathbf Z_\perp\|_2.
\label{eq:weyl_basic}
\end{equation}
Applying Lemma~\ref{lem:bern_Z} gives $\|\mathbf Z\|_2\le t(\eta)$ with probability at least $1-\eta$,
which proves the first claim.

For each $i\le r-1$, define $\hat g_i=\hat\lambda_i-\hat\lambda_{i+1}$ and
$g_i^\star=\lambda_i(\mathbf M)-\lambda_{i+1}(\mathbf M)$. Then
$$
\hat g_i-g_i^\star
=
(\hat\lambda_i-\lambda_i(\mathbf M))-(\hat\lambda_{i+1}-\lambda_{i+1}(\mathbf M)).
$$
Taking absolute values and using the triangle inequality together with~\eqref{eq:weyl_basic} gives
$$
|\hat g_i-g_i^\star|
\;\le\;
|\hat\lambda_i-\lambda_i(\mathbf M)| + |\hat\lambda_{i+1}-\lambda_{i+1}(\mathbf M)|
\;\le\;
2\|\mathbf Z_\perp\|_2
\;\le\;
2\|\mathbf Z\|_2.
$$
On the event $\{\|\mathbf Z\|_2\le t(\eta)\}$, this yields
$\max_{1\le i\le r-1}|\hat g_i-g_i^\star|\le 2t(\eta)$, proving the second claim.

Let $k^\star\in\arg\max_{1\le i\le r-1} g_i^\star$.
On the event $\{\max_{i\le r-1}|\hat g_i-g_i^\star|\le 2t(\eta)\}$, we have
$$
\hat g_{k^\star} \;\ge\; g_{k^\star}^\star - 2t(\eta),
\qquad
\hat g_i \;\le\; g_i^\star + 2t(\eta)\quad \forall\, i\neq k^\star.
$$
Therefore, if the assumed separation condition
$$
g_{k^\star}^\star - \max_{i\neq k^\star} g_i^\star \;>\; 4t(\eta)
$$
holds, then for all $i\neq k^\star$,
$$
\hat g_{k^\star}
\;\ge\;
g_{k^\star}^\star - 2t(\eta)
\;>\;
\max_{i\neq k^\star} g_i^\star + 2t(\eta)
\;\ge\;
\hat g_i,
$$
which implies $\hat k=k^\star$ (uniquely).
Since the event $\{\|\mathbf Z\|_2\le t(\eta)\}$ holds with probability at least $1-\eta$
by Lemma~\ref{lem:bern_Z}, the conclusion holds with probability at least $1-\eta$.
\end{proof}

\section{Difficulty of Shuffling for Power Method} \label{sec:difficulty}

We do not apply shuffle amplification to the noisy power method. Existing amplification results in the shuffle model~\citep[e.g.,][]{feldman2022hiding,feldman2023stronger} assume that each record generates one or more messages through a locally differentially private (LDP) randomizer before the shuffling step. Recall, that a randomized mechanism $R$ is $\varepsilon_0$-LDP if for all inputs $x,x'$ and all measurable sets $S$,
$$
\Pr[R(x)\in S] \le e^{\varepsilon_0}\Pr[R(x')\in S].
$$
In the noisy power iteration, however, privacy is introduced only after forming the global vector $A\theta_t$ by adding centrally calibrated noise, giving $y_t = A\theta_t + Z_t$, where $Z_t$ is the noise vector. Although one could subsequently permute (shuffle) the coordinates of $y_t$ prior to release, these coordinates are not outputs of per-record LDP randomizers. Consequently, the assumptions underlying classical shuffle amplification theorems are not satisfied, and the corresponding guarantees cannot be used.

More broadly, extending amplification results beyond mechanisms that decompose into independently privatized messages is known to be technically challenging. For instance, the numerical shuffle literature~\cite[see, e.g., Sec. 5,][]{koskela2023numerical} highlights the difficulty of establishing tight amplification bounds for general, data-dependent outputs in the shuffle model. 
For these reasons, we derive privacy guarantees for the noisy power method directly from the sensitivity of the matrix–vector multiplication together with the noise added at each iteration.

\newpage 

\section{Experiment on Private Gap Detection} \label{sec:exp_init}

\begin{figure*}[h!]
    \centering
    \includegraphics[width=0.43\textwidth]{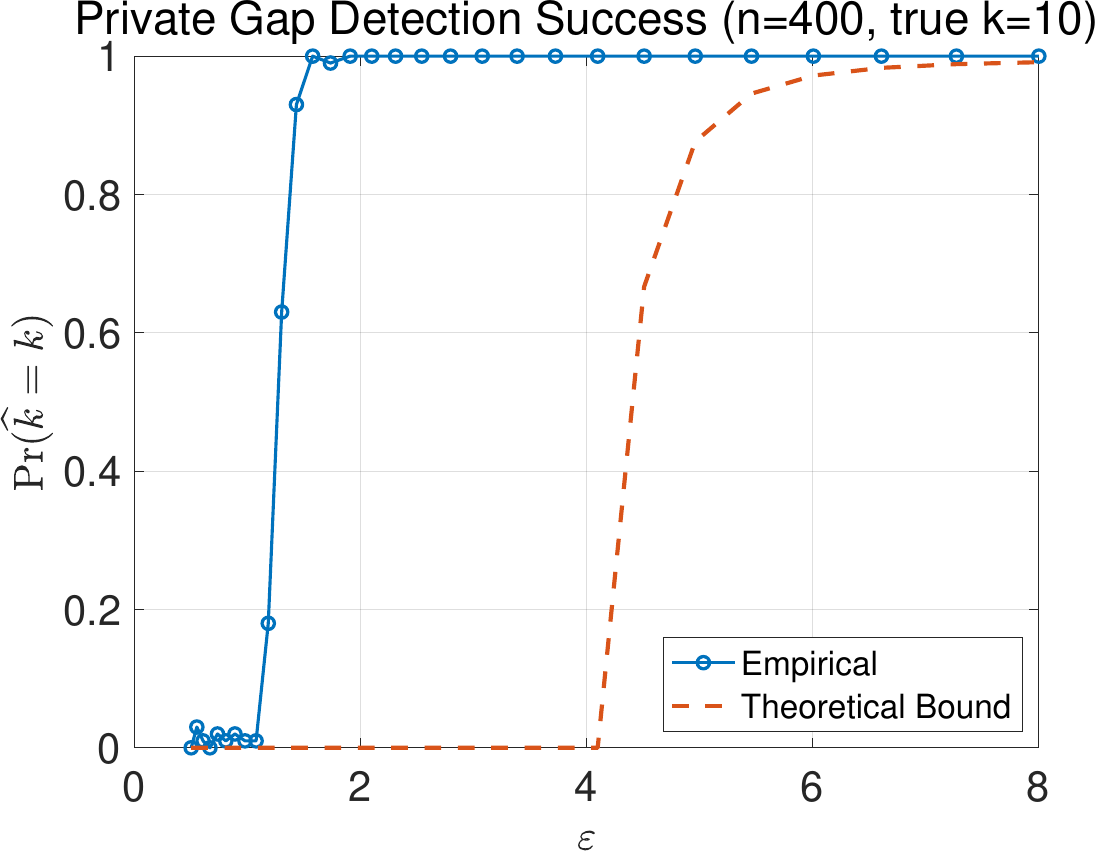}
    \caption{Empirical success probability of recovering the correct number of communities as a function of $\varepsilon$ under the matrix perturbation and shuffling mechanism. For each value of $\varepsilon$, we generate independent stochastic block model (SBM) graphs with $k=10$ equal-sized communities (each of size 50) and apply the projected eigengap initialization on the RR perturbed adjacency matrix. The success rate reports the fraction of trials in which the estimated number of communities $\widehat{k}$ matches the ground truth.}
    \label{fig:init}
\end{figure*}

\section{Further Experimental Comparison: Cora Citation Network Dataset}

We also consider the Cora citation network dataset~\citep{sen2008collective}, a widely used benchmark for graph-based classification. In this dataset, each node corresponds to a scientific publication  and edges represent citation links between papers.  Ground-truth labels divide the nodes into 7 different research topics.  We use the largest connected component, which contains $n = 2708$ nodes and $5429$ undirected edges. The resulting task is a seven-class node classification problem, with classes moderately imbalanced. As shown by Fig.~\ref{fig:cora}, the graph perturbation method shows again clearly the strongest performance.

\begin{figure} [h!]
     \centering
        \includegraphics[width= 0.43\columnwidth]{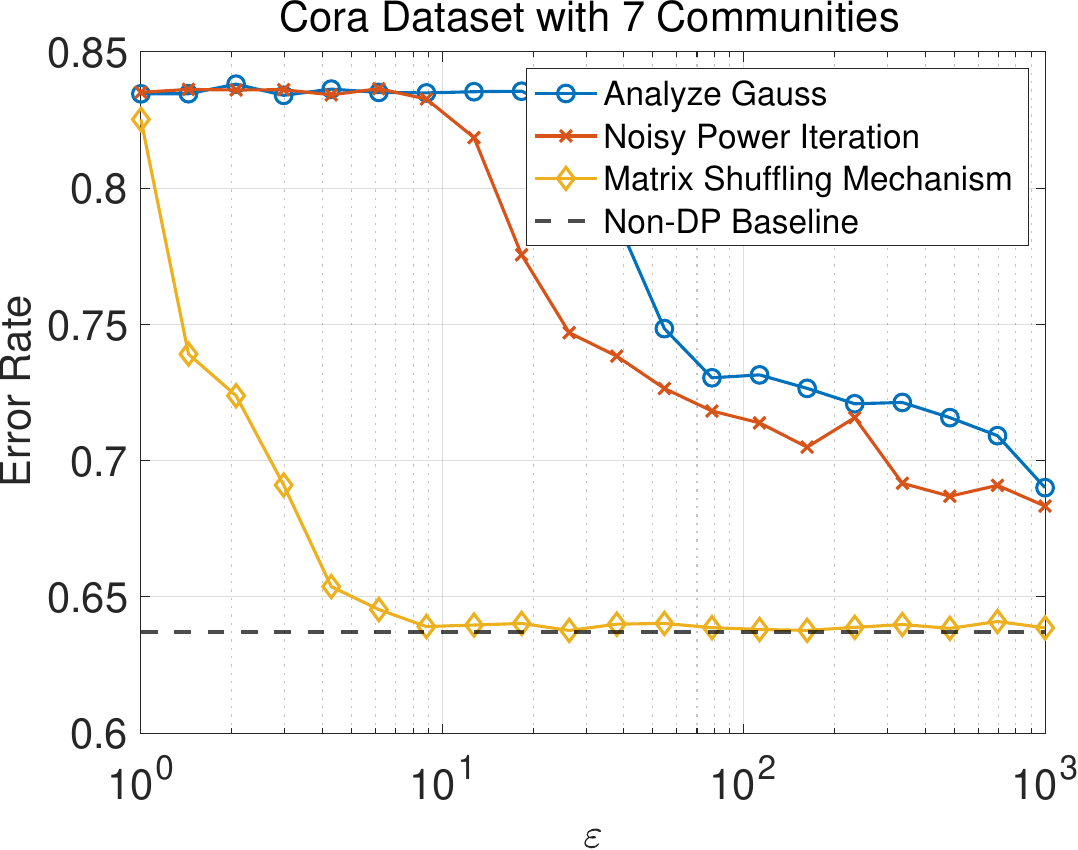}
        \caption{Error rate vs. $\varepsilon$-value, where $ k =7$ communities.}
        \label{fig:cora}
\end{figure}

\end{document}